\documentclass[12pt]{amsart}
\usepackage{geometry}
\usepackage{graphicx}
\usepackage{amssymb}
\usepackage{epstopdf}
\usepackage{amsmath,amscd}
\usepackage{amsthm}
\usepackage{url,verbatim}

\RequirePackage[colorlinks,citecolor=blue,urlcolor=blue]{hyperref}
\usepackage{breakurl}
\theoremstyle{plain}
\newtheorem{theorem}{Theorem}
\newtheorem{definition}[theorem]{Definition}
\newtheorem{lemma}[theorem]{Lemma}

\newtheorem{example}[theorem]{Example}
\newtheorem{assumption}[theorem]{Assumption}

\newcommand\ZZ{{\mathbb Z}}
\newcommand\NN{{\mathbb N}}

\renewcommand\ell{l}

\newcounter{mycount1}\newcounter{mycount2}\newcounter{mycount3}\newcounter{mycount}
\newenvironment{romlist}{\begin{list}{\rm(\roman{mycount1})}%
   {\usecounter{mycount1}\labelwidth=1cm\itemsep 0pt}}{\end{list}}
\newenvironment{numlist}{\begin{list}{\arabic{mycount2}.}%
   {\usecounter{mycount2}\labelwidth=1cm\itemsep 0pt}}{\end{list}}
\newenvironment{letlist}{\begin{list}{\rm(\alph{mycount3})}%
   {\usecounter{mycount3}\labelwidth=1cm\itemsep 0pt}}{\end{list}}
\newenvironment{Alist}{\begin{list}{\MakeUppercase{\alph{mycount}}.}%
   {\usecounter{mycount}\labelwidth=1cm\itemsep 0pt}}{\end{list}}

\numberwithin{equation}{section}
\numberwithin{theorem}{section}
\numberwithin{figure}{section}

\title[Positive speed self-avoiding walks]{Positive speed self-avoiding walks on graphs with more than one end}

\author{Zhongyang Li}

\address{Department of Mathematics,
University of Connecticut,
Storrs, Connecticut 06269-3009, USA} \email{zhongyang.li@uconn.edu}
\urladdr{\url{https://mathzhongyangli.wordpress.com}}

\begin{document}
\maketitle
\begin{abstract}
A self-avoiding walk (SAW) is a path on a graph that visits each vertex at most once. The mean square displacement of an $n$-step SAW is the expected value of the square of the distance between its ending point and starting point, where the expectation is taken with respect to the uniform measure on $n$-step SAWs starting from a fixed vertex. It is conjectured that the mean square displacement of an $n$-step SAW is asymptotically $n^{2\nu}$, where $\nu$ is a constant. Computing the exact values of the exponent $\nu$ on various graphs has been a challenging problem in mathematical and scientific research for long. 

 In this paper we show that on any locally finite Cayley graph of an infinite, finitely-generated group with more than two ends, the number of SAWs whose end-to-end distances are linear in lengths has the same exponential growth rate as the number of all the SAWs. We also prove that for any infinite, finitely-generated group with more than one end, there exists a locally finite Cayley graph on which SAWs have positive speed - this implies that the mean square displacement exponent $\nu=1$ on such graphs.  

These results are obtained by proving more general theorems for SAWs on quasi-transitive graphs with more than one end, which make use of a variation of Kesten's pattern theorem in a surprising way, as well as the Stalling's splitting theorem. Applications include proving that SAWs have positive speed on the square grid in an infinite cylinder, and on the infinite free product graph of two connected, quasi-transitive graphs.
\end{abstract}

\section{Introduction}

Self-avoiding walks, which are paths on graphs visiting no vertex more than once, were first introduced  as a model for long-chain polymers in chemistry (\cite{PF53}, see also \cite{MS96}). The theory of SAWs impinges on several areas of science including
combinatorics, probability, and statistical mechanics. Each of these
areas poses its characteristic questions concerning counting and geometry. Despite the simple definition, SAWs have been notoriously difficult to to study due to the fact that SAWs are, in general, non-Markovian.

The most natural SAW models are defined on regular graphs, such as the square-grid, the hexagonal lattice, etc; SAWs on these graphs have been studied extensively. In this paper, we consider SAWs on the general quasi-transitive graphs. Let $G=(V,E)$ be an infinite, connected graph, and let $\mathrm{Aut}(G)$ be the automorphism group for $G$. We say that $G$ is \textbf{quasi-transitive}, if there exists a subgroup $\Gamma$ of $\mathrm{Aut}(G)$ acting quasi-transitively on $G$, i.e. the action of $\Gamma$ on $V$ has only finitely many orbits. More precisely, there exist a finite set of vertices $W\subset V$, $|W|<\infty$, such that for any $v\in V$, there exist $w\in W$ and $\gamma\in\Gamma$ with $w=\gamma v$. The set $W$ is called a \textbf{fundamental domain}. A graph is called \textbf{locally finite}, if every vertex has finite degree, i.e., incident to finitely many edges. A subset of vertices $U\subseteq V$ is called \textbf{connected}, if for any $p,q\in U$, there exists $u_0(=p),\ u_1,\ \ldots,\ u_{n-1},\ u_n(=q)\in U$ such that for $1\leq i\leq n$, $u_{i-1}$ and $u_i$ are adjacent vertices (two vertices joined by an edge).

The \textbf{connective constant} is a fundamental quantity concerning counting the number of SAWs starting from a fixed vertex, and this is the starting point of a rich theory of geometry and phase transition. It is defined, on a quasi-transitive graph, to be the exponential growth rate of the number of $n$-step SAWs starting from a fixed vertex. More precisely, let $c_n(v)$ be the number of $n$-step SAWs starting from a fixed vertex $v\in V$, the connective constant $\mu$ is defined to be
\begin{eqnarray}
\mu:=\lim_{n\rightarrow\infty} [\mathrm{sup}_{v\in V}c_n(v)]^{\frac{1}{n}}\label{cc}
\end{eqnarray}
The limit on the right hand side of (\ref{cc}) is known to exist by a sub-additivity argument. It is proved in \cite{jmhII} that the connective constant $\mu$ defined in (\ref{cc}), can be expressed as follows
\begin{eqnarray}
\mu=\lim_{n\rightarrow\infty} c_n(v)^{\frac{1}{n}},\qquad\forall v\in V.\label{cc1}
\end{eqnarray}

Although the definition of the SAW is quite simple, a lot of fundamental questions concerning SAWs remain unknown. For example, it is still an open problem to compute the exact value of the connective constant for the 2-dimensional square grid. A recent breakthrough is a proof of the fact that the connective constant of the hexagonal lattice is $\sqrt{2+\sqrt{2}}$; see \cite{ds}. See \cite{GL-Comb,GrL3,GL-amen,GLgmn,GrLrev} for results concerning bounds of connective constants on quasi-transitive graphs; \cite{GL-loc,GL-cayley} for results concerning the dependence of connective constants on local structures of graphs; \cite{GLWSAW} for the continuous dependence of connective constants of weighted SAWs on edge weights of the graph; and \cite{GLFs} for the changes of the connective constant of SAWs under local transformations of the underlying graph.

Another important quantity relating to SAWs is the \textbf{mean square displacement exponent $\nu$}. Let $\pi_n^v$ be an $n$-step SAW on $G$ starting from a fixed vertex $v$, and let 
\begin{eqnarray*}
\|\pi_n^v\|=\mathrm{dist}_{G}(\pi(n),\pi(0)),
\end{eqnarray*}
where $\mathrm{dist}_{G}(\cdot,\cdot)$ is the graph distance on $G$. Let $\langle\cdot \rangle$ be the expectation taken with respect to the uniform probability measure for $n$-step SAWs on $G$ starting from a fixed vertex.  The mean square displacement exponent $\nu$ for SAWs, defined by
\begin{eqnarray*}
\langle\|\pi_n^v\|^2 \rangle\sim n^{2\nu},
\end{eqnarray*}
has been an interesting topic to mathematicians and scientists for long. Here ``$\sim$" means that there exist constants $C_1,C_2>0$, independent of $n$, such that $C_1 n^{2\nu}\leq \langle\|\pi_n^v\|^2 \rangle\leq C_2n^{2\nu}$. 

Although the connective constant depends on the local structure of the graph, the mean square displacement exponent $\nu$ is believed to be universal in the sense that it depends only on the dimension of the space where the graph is embedded, but independent of the graph.  It is conjectured that $\nu=\frac{3}{4}$ for SAWs on graphs embedded in the 2-dimensional Euclidean plane (in particular, this means that the square grid, the hexagonal lattice and the triangular lattice share the same exponent $\nu=\frac{3}{4}$, although they obviously have distinct connective constants),  $\nu=\frac{1}{2}$ for SAWs on $\ZZ^d$ with $d\geq 4$, and that $\nu=1$ for SAWs on a non-amenable graph with bounded vertex degree. See \cite{GLFs} for the invariance of SAW components under local transformations of cubic (valent-3) graphs.

The conjecture that $\nu=\frac{1}{2}$ when $d\geq 5$ was proved in \cite{BS85,HS92}. See \cite{bdgs} for related results when $d=4$, and \cite{CH13} for related results for $d\geq 2$.

It is proved in \cite{NP12} that if a non-amenable Cayley graph satisfies 
\begin{eqnarray}
(\Delta-1)\rho\mu^{-1}<1,\label{cd12}
\end{eqnarray}
 then SAWs have positive speed. Here $\Delta$ is the vertex degree, $\rho$ is the spectral radius for the transition matrix of the simple random walk on the graph, and $\mu$ is the connective constant as defined in (\ref{cc}). Combining with the results in \cite{BM88,BS96,PS00}, it is known that for any finitely generated non-amenable group, there exists a locally finite Cayley graph on which SAWs have positive speed.

It is proved in \cite{MW05} that SAWs have positive speed for certain regular tilings of the hyperbolic plane. An upper bound of the spectral radius for a planar graph with given maximal degree is proved in \cite{DM10}, which, combining with (\ref{cd12}), can be used to show that SAWs have positive speed on a large class of planar graphs. It is shown in \cite{ib16} that SAWs on the 7-regular infinite planar triangulation has linear expected displacement.

The main goal of this paper is to study the mean square displacement exponent $\nu$ for SAWs on quasi-transitive graphs with more than one end.
The number of \textbf{ends} of a connected graph is the supremum over its finite subgraphs of the number of infinite components that remain after removing the subgraph.

Let $G=(V,E)$ be an infinite, connected, locally finite, quasi-transitive graph with more than one end. Let $\Gamma\subseteq \mathrm{Aut}(G)$ be a subgroup of the automorphism group of $G$ acting quasi-transitively on $G$. Since $G$ has more than one end, there exists a finite subset of $V$ (which is called a ``cut set"), such that after removing all the vertices as well as incident edges of the set, the remaining graph has at least two infinite components. If distinct components of the remaining graph have certain ``symmetry" under the action of $\Gamma$, one may map certain portions of an SAW from one component to another component of the remaining graph and form a new SAW, such that the end-to-end distance of the new SAW is linear in its length. Then the number of $n$-step SAWs with end-to-end distance linear in $n$, when $n$ is large, may be compared with the total number of $n$-step SAWs. To that end, we may make the following assumptions on the graph $G$ concerning the ``symmetry'' of different components after removing the finite ``cut set''.

\begin{assumption}\label{ap31}There exist a finite set of vertices $S$, $S\subset V$ and $|S|<\infty$, such that 
\begin{enumerate}
\item $S$ is connected;
\item  $G\setminus S$ (the graph obtained from $G$ by removing all the vertices in $S$ and their incident edges) has at least two infinite components;
\item for any component $A$ of $G\setminus S$, let $\partial_A S$ be the set consisting of all the vertices in $S$ incident to a vertex in $A$. There exists an infinite component $B$ of $G\setminus S$ and a graph automorphism $\gamma\in\Gamma$, such that $B\cap A=\emptyset$, $\gamma A\subseteq B$; for any $v\in \partial_A S$, $\gamma v\in \partial_B S\cup B$, $v$ and $\gamma v$ are joined by a path in $G\setminus(A\cup\gamma A)$, whose length is bounded above by a constant $N$ independent of $A,v$. Denote $\gamma$ by $\phi(S,A):=\gamma$.
\end{enumerate}
See Figure \ref{fig:ap31}.
\end{assumption}

\begin{figure}[hb]
  \centering
  \includegraphics[width=15cm]{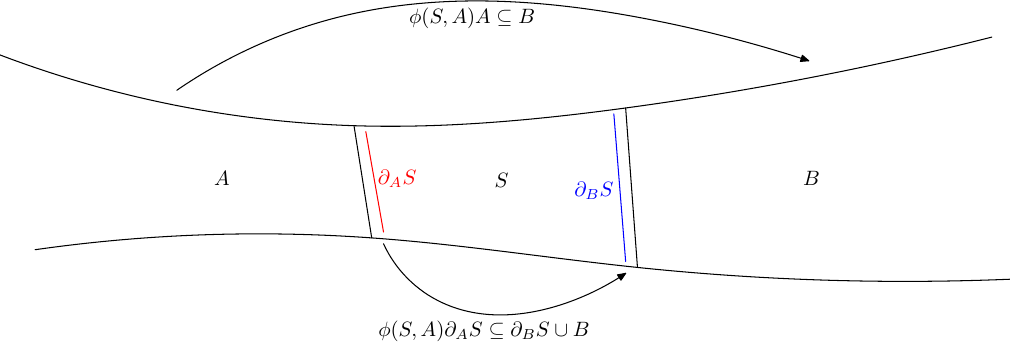}
  \caption{Assumption \ref{ap31}}\label{fig:ap31}
\end{figure}

\begin{assumption}\label{ap32}There exist a finite set of vertices $S$, $S\subset V$ and $|S|<\infty$ satisfying Assumption \ref{ap31}. Moreover, assume that
\begin{itemize}
\item there exist a finite set of vertices $S'$, such that $S\subseteq S'$. Let $\partial S'$ be the set consisting of all the vertices in $S'$ incident to a vertex in $G\setminus S'$. For any two distinct vertices $u,v\in\partial S'$, there exists an SAW $\ell_{uv}$ joining $u$ and $v$ and visiting every vertex in $S$.
\end{itemize}
See Figure \ref{fig:ap32}.
\end{assumption}

\begin{figure}[hb]
  \centering
  \includegraphics[width=15cm]{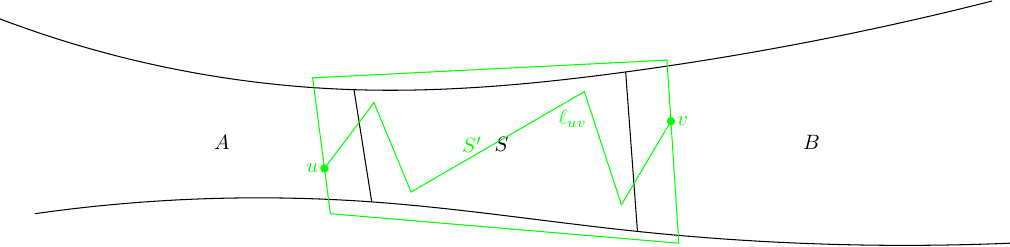}
  \caption{Assumption \ref{ap32}}\label{fig:ap32}
\end{figure}

Here are the main results of the paper.

\begin{theorem}\label{m31}Let $G=(V,E)$ be an infinite, connected, locally finite, quasi-transitive graph with more than one end. Let $\mu$ be the connective constant of $G$. Let $\pi_n^v$ be an $n$-step SAW on $G$ starting from a fixed vertex $v$.
\begin{Alist} 
\item If $G$ satisfies Assumption \ref{ap31}, then there exists $a\in(0,1]$
\begin{eqnarray*}
\limsup_{n\rightarrow\infty} \sup_{v\in V}|\{\pi_n^v:\|\pi_n^v\|\geq a n\}|^{\frac{1}{n}}=\mu.
\end{eqnarray*}
\item If $G$ satisfies Assumption \ref{ap32}, then $\pi_n^v$ has positive speed, i.e., there exist constants $C,\alpha,\beta>0$, such that
\begin{eqnarray*}
\mathbb{P}_n(\|\pi_n^v\| \leq \alpha n)\leq Ce^{-n\beta}.
\end{eqnarray*}
where $\mathbb{P}_n$ is the uniform measure on the set of $n$-step SAWs on $G$ starting from a fixed vertex.
\end{Alist}
\end{theorem}

For a graph satisfying Assumption \ref{ap32}, Theorem \ref{m31} implies that the mean square displacement of SAWs on the graph is of the order $n^2$, i.e.
\begin{eqnarray*}
\langle\|\pi_n^v\|^2 \rangle\sim n^2.
\end{eqnarray*}

The approach to prove Theorem \ref{m31} is to consider a finite ``cut set'' $S$ as given by Assumption \ref{ap31}, such that SAWs, once crossing this ``cut set'', will move to another component of $G\setminus S$ and most of them may never come back again. The analysis involves arguments and technical details inspired by the pattern theorem (\cite{HK63}), see also (\cite{MS96,CH13,GrL3,SW90,ASW09}). The proofs of Part A. and Part B. are similar; note that under the stronger assumption \ref{ap32}, not only the the number of $n$-step SAWs whose end-to-end distance is linear in $n$ has the same exponential growth rate as the total number of $n$-step SAWs starting from a fixed vertex, but the number of of $n$-step SAWs whose end-to-end distance is not linear in $n$ is actually exponential small compared to the total number of $n$-step SAWs starting from a fixed vertex.

Applications of Theorem \ref{m31} include a proof that SAWs on an infinite cylindrical square grid have positive speed, and that SAWs on an infinite free product graph of two quasi-transitive, connected graphs have positive speed. 

\begin{example}(Cylinder) Consider the quotient graph of the square grid $\ZZ^2$, $\ZZ\times \ZZ_{\ell}$, where $\ell$ is a positive integer. This is a graph with two ends. We can choose $S=\{0\}\times\ZZ_{\ell}$ and $S'=\{-1,0,1\}\times Z$. Then Assumption \ref{ap32} is satisfied and SAWs have positive speed. See also \cite{FSG99} for discussions about SAWs on a cylinder.
\end{example}

\begin{definition}(Free product of graphs)\label{df41} Let $G_1=(V_1,E_1,o_1)$, $G_2=(V_2,E_2,o_2)$ be two connected, locally finite, quasi-transitive, rooted graphs with vertex sets $V_1$, $V_2$; edge sets $E_1,E_2$ and roots $o_1\in V_1,o_2\in V_2$, respectively.  For $i\in\{1,2\}$, assume that 
\begin{enumerate}
\item $|V_i|\geq 2$; and
\item $V_i^{\times}=V_i\setminus\{o_i\}$; and
\item $I(x)=i$ if $x\in V_i^{\times}$. 
\end{enumerate}
Define
\begin{eqnarray*}
V:=V_1*V_2=\{x_1x_2\ldots x_n|n\in \mathbb{N},x_k\in V_1^{\times}\cup V_2^{\times}, I(x_k)\neq I(x_{k+1})\}\cup\{o\}
\end{eqnarray*}
We define an edge set $E$ for the vertex set $V$ as follows: if $i\in\{1,2\}$ and $x,y\in V_i$, and $(x,y)\in E_i$, then $(wx,wy)\in E$ for all $w\in V$. See \cite{GM15} for discussions of SAWs on free product graphs of quasi-transitive graphs.
\end{definition}

\begin{theorem}\label{tm51}Let $G=(V,E)$ be the free product graph of two connected, locally finite, quasi-transitive, rooted graphs $G_1=(V_1,E_1,o_1)$ and $G_2=(V_2,E_2,o_2)$ with $|V_i|\geq 2$, for $i=1,2$, as defined in \ref{df41}. Then SAWs on $G$ have positive speed.
\end{theorem}

The \textbf{ends of a finitely generated group} are defined to be the ends of the corresponding Cayley graph; this definition is insensitive to the choice of the finite generating set. It is well known that every finite-generated infinite group has either 1, 2, or infinitely many ends. Concerning groups with more than one end, 
Theorem \ref{m31} also has the following corollaries.

\begin{theorem}\label{mg}Let $\Gamma$ be an infinite, finitely-generated group with more than two ends. Let $G=(V,E)$ be a locally finite Cayley graph of $\Gamma$. For $v\in V$ Let $\pi_n^v$ be an $n$-step SAW on $G$ starting from $v$. Then 
\begin{eqnarray*}
\limsup_{n\rightarrow\infty}|\{\pi_n^v:\|\pi_n^v\|\geq a n\}|^{\frac{1}{n}}=\mu.
\end{eqnarray*}
\end{theorem}

\begin{theorem}\label{t15}Let $\Gamma$ be an infinite, finitely-generated group with more than one end. There exists a locally finite Cayley graph $G=(V,E)$ of $\Gamma$, such that SAWs on $G$ have positive speed.
\end{theorem}

The proofs of Theorems \ref{mg} and \ref{t15} make use of the Stalling's splitting theorem (see \cite{JS68}), which gives explicit presentations for groups with more than one end; as well as constructions of sets $S$ and $S'$ satisfying Assumptions \ref{ap31} and \ref{ap32}.

The organization of the paper is as follows. In Section \ref{pp1}, we prove Theorem \ref{m31} A. In Section \ref{ppb}, we prove Theorem \ref{m31} B. Theorems \ref{mg} and \ref{t15} are proved in Section \ref{p2}. In Section \ref{fp}, we prove Theorem \ref{tm51}.

\section{Proof of Theorem \ref{m31} A.} \label{pp1}

This section is devoted to prove Theorem \ref{m31} A.

Let $G=(V,E)$ be a graph satisfying the assumption of Theorem \ref{m31}.
Let $S$ be a finite set of vertices satisfying Assumption \ref{ap31}. Recall that $\Gamma$ is a subset of $\mathrm{Aut}(G)$ acting quasi-transitively on $G$. Let $\Gamma S$ be the set of images of $S$ under $\Gamma$. By quasi-transitivity of $G$, for each $\gamma\in \Gamma$, $\gamma S$ still satisfies Assumption \ref{ap31}. 

We shall next introduce events $E^*$, $E_k$ and $\tilde{E}_k$ and their restrictions to a length-$2m$ sub-walks $E^*(m)$, $E_k(m)$ and $\tilde{E}_k(m)$, where $m$ is a fixed postive integer. In the proof of Theorem \ref{m31}, we shall modify an $n$-step SAW to a new SAW such that these events appear at least $an$ times in the new SAW; for some $0<a<1$; moreover, we may choose $\delta n$ of the $an$ occurrences of these events for some $0<\delta<a$ and map part of the SAW there by a graph automorphism to a different component of the remaining graph after removing the ``cut set'', this way, we construct a new SAW whose end-to-end distance is linear in $n$ since it crosses the ``cut set'' at least $\delta n$ times. Different choices of locations of these events for the modifications and mappings to construct new SAWs will give an exponential factor strictly greater than 1 on the number of SAWs whose end-to-end distances are linear in $n$, compared to the number of those SAWs whose end-to-end distances are not.

Let $\pi$ be an $n$-step SAW on $G$. We say that $E^*$ occurs at the $j$th step of $\pi$ if there exists $\gamma S\in \Gamma S$ such that $\pi(j)\in \gamma S$, and all the vertices of $\gamma S$ are visited by $\pi$. For $k\geq 1$, we say that $E_k$ occurs at the $j$th step of $\pi$, if there exists $\gamma S\in \Gamma S$, such that $\pi(j)\in \gamma S$, and at least $k$ vertices of $\gamma S$ are visited by $\pi$. We say that $\tilde{E}_k$ occurs at the $j$th step of $\pi$ if $E^*$ or $E_k$ (or both) occur there.

In the following, we will use $E$ to denote any of $E^*$, $E_k$ or $\tilde{E}_k$. If $m$ is a positive integer, we say that $E(m)$ occurs at the $j$th step of $\pi$ if $E$ occurs at the $m$th step of the $2m$-step subwalk $(\pi(j-m),\ldots,\pi(j+m))$. (If $j-m<0$ or $j+m>n$, then an obvious modification must be made in this definition: for $j-m<0$, it means that $E$ occurs at the $j$th step of $(\pi(0),\ldots, \pi(j+m))$; for $j+m>n$, it means that $E$ occurs at the $m$th step of $(\pi(j-m), \ldots, \pi(n))$). In particular, if $E(m)$ occurs at the $j$th step of $\pi$, then $E$ occurs at the $j$th step of $\pi$.

Let $c_n(v)$ be the number of $n$-step SAWs on $G$ starting from a fixed vertex $v$. For $r\geq 0$, let $c_n^v(r,E)$ (resp.\ $c_n^v(r,E(m))$) be the number of $n$-step SAWs starting from $v$ for which $E$ (resp.\ $E(m)$) occurs at no more than $r$ different  steps. 

Let 
\begin{eqnarray*}
c_n&=&\mathrm{sup}_{v\in V}c_n(v);\\
c_n(r,E)&=&\mathrm{sup}_{v\in V}c_n^v(r,E);\\
\lambda(E)&=&\limsup_{n\rightarrow\infty}c_n(0,E)^{\frac{1}{n}}.
\end{eqnarray*}
Let $\mu$ be the connective constant of $G$ as defined in (\ref{cc}). We have that $\lambda(E)<\mu$ if and only if 
\begin{eqnarray}
\mathrm{there\ exist\ }\epsilon>0, M\in\NN, \mathrm{such\ that\ } c_m(0,E)<[\mu(1-\epsilon)]^m, \mathrm{for\ } m\geq M.\label{em1}
\end{eqnarray}

\begin{lemma}\label{l1}Suppose that
\begin{eqnarray}
\lambda(E)<\mu.\label{lk}
\end{eqnarray}
Let $\epsilon,M$ satisfy (\ref{em1}), and let $m\geq M$ satisfy
\begin{eqnarray}
c_m\leq [\mu(1+\epsilon)]^m.\label{em2}
\end{eqnarray}
Then there exist $a=a(\epsilon,m)$ and $R=R(\epsilon,m)\in (0,1)$, such that
\begin{eqnarray*}
\limsup_{n\rightarrow\infty}c_n(an,E(m))^{\frac{1}{n}}<R\mu.
\end{eqnarray*}
\end{lemma}

\begin{proof}Assume that (\ref{lk}) holds. Let $\epsilon,m$ satisfy (\ref{em2}). Since $c_m(0,E)=c_m(0,E(m))$, by (\ref{em1}) we have
\begin{eqnarray*}
c_m(0,E(m))<[\mu(1-\epsilon)]^m.
\end{eqnarray*}
Let $\pi$ be an $n$-step SAW on $G$ and $L=\left\lfloor\frac{n}{m}\right\rfloor$. If $E(m)$ occurs in no more than $r$ steps of $\pi$, then $E(m)$ occurs at no more than $r$ of the $m$-step subwalks
\begin{eqnarray*}
(\pi((j-1)m),\pi((j-1)m+1),\ldots,\pi(jm)),\qquad 1\leq j\leq L.
\end{eqnarray*}

Counting the number of ways in which $r$ or fewer of these subwalks can contain an occurrence of $E(m)$, we have that
\begin{eqnarray*}
c_n(r,E(m))&\leq& \sum_{i=1}^{r}\left(\begin{array}{c}L\\i\end{array}\right)c_m^i[c_m(0,E(m))]^{L-i}c_{n-Lm}\\
&\leq &\mu^{Lm}c_{n-Lm}\sum_{i=1}^{r}\left(\begin{array}{c}L\\i\end{array}\right)(1+\epsilon)^{im}(1-\epsilon)^{(L-i)m}
\end{eqnarray*}

For $\xi$ small and positive, we have
\begin{eqnarray*}
&&\sum_{i=0}^{\xi L}\left(\begin{array}{c}L\\i\end{array}\right)(1+\epsilon)^{im}(1-\epsilon)^{(L-i) m}\\
&\leq &(\xi L+1)\left(\begin{array}{c}L\\ \xi L\end{array}\right)\left(\frac{1+\epsilon}{1-\epsilon}\right)^{\xi Lm}(1-\epsilon)^{Lm}
\end{eqnarray*}

The $L$th root of the right hand side converges as $L\rightarrow\infty$ to
\begin{eqnarray*}
f(\xi)=\frac{1}{\xi^{\xi}(1-\xi)^{1-\xi}}\left(\frac{1+\epsilon}{1-\epsilon}\right)^{\xi m}(1-\epsilon)^m,
\end{eqnarray*}
which is strictly less than 1 for $0<\xi<\xi_0$, and some $\xi_0=\xi_0(\epsilon,m)>0$. Therefore when $0<a<\frac{\xi}{m}$, and $R=f(\xi)^{\frac{1}{m}}$, we have
\begin{eqnarray*}
c_n(an, E(m))^{\frac{1}{n}}<R\mu,
\end{eqnarray*}
and the proof is complete.
\end{proof}

\begin{lemma}\label{qtc}Let $G_1=(V(G_1),E(G_1))$ be a component of $G\setminus \Gamma S$. There exists a subgroup $\Gamma_1$ of $\Gamma$ acting quasi-transitively on $G_1$.
\end{lemma}

\begin{proof}Since $\Gamma$ acts on $G$ quasi-transitively, $V$ has finitely many orbits under the action of $\Gamma$. Let $W_1$ be the subset of $V(G_1)$ consisting of one representative in each orbit of $V$ under the action of $\Gamma$, such that the intersection of the orbit with $V(G_1)$ is nonempty, then $|W_1|<\infty$.

Let 
\begin{eqnarray*}
\Gamma_1=\{\gamma\in\Gamma:\forall w\in W_1,\gamma w\in V(G_1)\}.
\end{eqnarray*}
Then it is straightforward to check that $\Gamma_1$ is a subgroup of $\Gamma$, and that $\Gamma_1$ acts on $G_1$ quasi-transitively.
\end{proof}

\begin{lemma}\label{lcc}There exists a component $G_1=(V(G_1),E(G_1))$ of $G\setminus\Gamma S$, such that $\lambda(\tilde{E}_1)$ is the connective constant of $G_1$, i.e.
\begin{eqnarray}
\lambda(\tilde{E}_1)&=&\lim_{n\rightarrow\infty} \sup_{v\in V(G_1)}\tilde{c}_n(v)^{\frac{1}{n}}\label{d1}\\
&=&\lim_{n\rightarrow\infty}\tilde{c}_n(v)^{\frac{1}{n}},\ \forall v\in V(G_1)\label{d2}
\end{eqnarray}
where $\tilde{c}_n(v)$ is the number of $n$-step SAWs on $G_1$ starting from $v$.
\end{lemma}
\begin{proof}By definition of $\lambda(\tilde{E}_1)$, we have
\begin{eqnarray*}
\lambda(\tilde{E}_1)=\lim_{n\rightarrow\infty}\sup_{v\in V\setminus \Gamma S}\overline{c}_n(v)^{\frac{1}{n}},
\end{eqnarray*}
where $\overline{c}_n(v)$ is the number of $n$-step SAWs starting from $v$ on the component of $G\setminus \Gamma S$ including $v$.

By the quasi-transitivity of $G\setminus \Gamma S$ under the action of $\Gamma$, there exists $v_1\in G\setminus\Gamma S$, such that 
\begin{eqnarray*}
\lambda(\tilde{E}_1)=\limsup_{n\rightarrow\infty} \overline{c}_n(v_1)^{\frac{1}{n}}.
\end{eqnarray*}
Let $G_1$ be the component of $G\setminus \Gamma S$ containing $v_1$, then
\begin{eqnarray*}
\lambda(\tilde{E}_1)=\limsup_{n\rightarrow\infty}\overline{c}_n(v_1)^{\frac{1}{n}}\leq\lim_{n\rightarrow\infty}\sup_{v\in V(G_1)}\overline{c}_n(v)^{\frac{1}{n}}\leq \lim_{n\rightarrow\infty} \sup_{v\in V\setminus \Gamma S}\overline{c}_n(v)^{\frac{1}{n}}=\lambda(\tilde{E}_1).
\end{eqnarray*}
Therefore $\lambda(\tilde{E}_1)$ is the connective constant for $G_1$, and (\ref{d1}) follows. The identity (\ref{d2}) follows from Lemma \ref{qtc} and (\ref{cc}), (\ref{cc1}).
\end{proof}

\begin{lemma}\label{l2}$\lambda(\tilde{E}_1)<\mu$.
\end{lemma}
\begin{proof}The main idea to prove the lemma is to ``lift" an SAW on a component to $G\setminus \Gamma S$ to SAWs on $G$; there are multiple ways of doing this. Different ways of lifting one SAW on a component of $G\setminus \Gamma S$ to multiple SAWs on $G$ will give a nontrivial exponential factor on the number of SAWs on these two graphs; and therefore a strict inequality on the corresponding connective constants.

Since $\Gamma$ acts on $G$ quasi-transitively, let $W$ be a fundamental domain such that $|W|<\infty$. Let 
\begin{eqnarray*}
N_0=\max_{w_1,w_2\in W}\mathrm{dist}_{G}(w_1,w_2).
\end{eqnarray*}
For any $v\in V$, let 
\begin{eqnarray*}
B_G(v,N_0)=\{u\in V:\mathrm{dist}_{G}(v,u)\leq N_0\}.
\end{eqnarray*}
It is not hard to see that for any $v\in V$, $B_G(v,N_0)$ contains a vertex in each orbit of $V$ under the action of $\Gamma$.
Therefore for any $v\in V$, there exists $\gamma S\in \Gamma S$, such that 
\begin{eqnarray*}
B_G(v,N_0)\cap \gamma S\neq \emptyset.
\end{eqnarray*}

Let $G_1$ be a component of $G\setminus\Gamma S$, such that $\lambda(\tilde{E}_1)$ is the connective constant of $G_1$. The existence of $G_1$ is guaranteed by Lemma \ref{lcc}.
Let $T_n(v)$ be the set of all $n$-step SAWs on $G_1$ starting from the vertex $v$. For each $\pi\in T_n(v)$, find indices $j_1,\ldots, j_u$, such that for $i\neq \ell$, we have
\begin{eqnarray}
\mathrm{dist}_{G_1}(\pi(j_i),\pi(j_{\ell}))\geq 8(N_0+|S|)+N,\label{a}
\end{eqnarray}
here $N$ is given by Assumption \ref{ap31} (3).

We may assume that $u=\kappa n$. For each $B_{G}(\pi(j_i), N_0)$, there exists $\gamma_{j_i}S\in\Gamma S$, such that 
\begin{eqnarray}
B_G(\pi(j_i),N_0)\cap \gamma_{j_i} S\neq \emptyset.\label{ine}
\end{eqnarray}
Let $\pi(t_i)$ be a closest vertex on $\pi$, in graph distance of $G$,  to $\gamma_{j_i}S$. By (\ref{ine}), we have
\begin{eqnarray}
\mathrm{dist}_{G}(\pi(t_i),\gamma_{j_i}S)\leq N_0.\label{dln}
\end{eqnarray}
 Then
\begin{eqnarray}
\mathrm{dist}_{G}(\pi(j_i),\pi(t_i))&\leq& \mathrm{dist}_{G}(\pi(j_i),\gamma_{j_i}S)+\mathrm{dist}_{G}(\gamma_{j_i}S,\pi(t_i))+|S|\label{b}
\\&\leq& 2N_0+|S|.\notag
\end{eqnarray}
By rearrangements if necessary, we may assume that 
\begin{eqnarray*}
t_1<t_2<\ldots<t_{\kappa n}.
\end{eqnarray*}

We choose a subset of indices  $H\subset \{t_1,t_2,\ldots,t_{\kappa n}\}$ to perform a manipulation, which will be described later. Assume that $H=\{h_1,\ldots, h_{\delta n}\}$, where $0<\delta<\kappa$, and $h_a<h_b$ if $a<b$. For each $h_{\ell}\in H$ with $h_{\ell}=t_i$, let $S_{\ell}=\gamma_{j_i}S$.

We construct a new SAW $\pi_1$ as follows
\begin{itemize}
\item let $\alpha_1$ be the subwalk of $\pi$ from $\pi(0)$ to $\pi(h_1)$;
\item use a shortest path to join $\pi(h_1)$ and $S_1$, denoted by $\omega_1$;
\item let $A_1$ be the component of $G\setminus S_1$ containing $\pi(h_1)$; use $\phi(S_1,A_1)$ to map the concatenation of the reversed $\omega_1$ and $\pi(h_1),\pi(h_1+1),\ldots,\pi(n)$ to another component of $G\setminus S_1$, and denote the image of the concatenation of the two subwalks under $\phi(S_1,A_1)$ by $\theta_1$ - Here $\phi(S_1,A_1)$ is given as in Assumption \ref{ap31} (3);
\item use an SAW $\beta_1$ in $G\setminus[A_1\cup\phi(S_1,A_1)A_1]$ to join the last vertex of $\omega_1$ and its image under $\phi(S_1,A_1)$ (note that the existence of $\beta_1$ is guaranteed by Assumption \ref{ap31});
\item let $\pi_1$ be the concatenation of $\alpha_1$, $\omega_1$, $\beta_1$ and $\theta_1$.
\end{itemize}
It is not hard to check that $\pi_1$ is indeed an SAW. See Figure \ref{fig:rel}.
  \begin{figure}[hb]
  \centering
  \includegraphics[width=15cm]{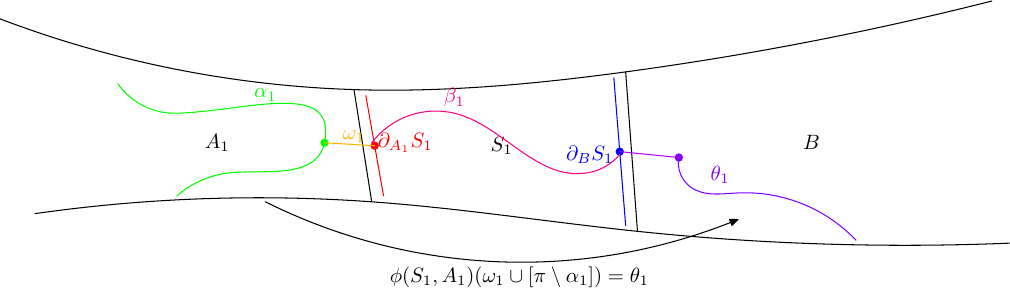}
  \caption{}\label{fig:rel}
\end{figure}
Note that $\alpha_1$ is a subwalk of both $\pi$ and $\pi_1$; and the subwalk $\pi\setminus \alpha_1$ is mapped, by $\phi(S_1,A_1)$ to a subwalk of $\theta_1$. Let $h_1^1=h_1$, and let $h_i^1\ (2\leq i\leq \delta n)$ be the step in $\pi_1$ of the image of $\pi(h_i)$ under $\phi(S_1,A_1)$; i.e.,
\begin{eqnarray}
\phi(S_1,A_1)[\pi(h_i)]=\pi_1(h_i^1).\label{img}
\end{eqnarray}

We now prove the following lemma
\begin{lemma}\label{ll25}
For $2\leq i\leq \delta n$, $\pi_1(h_i^1)$ is a closest vertex on $\pi_1$ to $\phi(S_1,A_1)S_i$. 
\end{lemma}
\begin{proof}
First of all, since $\pi_1(h_2^1-h_2+h_1),\ldots,\pi_1(n+h_2^1-h_2)$ is the image of $\pi(h_1),\ldots, \pi(n)$ under $\phi(S_1,A_1)$, we have that for $2\leq i\leq \delta n$, $\pi_1(h_i^1)$ is a closest vertex on the subwalk $\pi_1(h_1^1-h_2+h_1),\ldots,\pi_1(n-h_2+h_2^1)$ to $\phi(S_1,A_1)S_i$. Moreover, by (\ref{dln}) and (\ref{img}), we have
\begin{eqnarray*}
\mathrm{dist}_{G}(\pi_1(h_i^1),\phi(S_1,A_1)S_i)\leq N_0, \ \mathrm{for}\ 2\leq i\leq \delta n.
\end{eqnarray*}
However, for any vertex $v$ on the subwalk $\pi(0),\ldots, \pi_1(h_2^1-h_2+h_1)$, we have
\begin{eqnarray}
\mathrm{dist}_{G}(v,\phi(S_1,A_1)S_i)> N_0, \ \mathrm{for}\ 2\leq i\leq \delta n.\label{c}
\end{eqnarray}
To see why (\ref{c}) is true, assume $S_i=\gamma_{j_k}S$ and $h_1=t_l$; note that from (\ref{a}), (\ref{ine}), (\ref{b}), we have
\begin{eqnarray*}
&&\mathrm{dist}_{G}(\pi_1(h_2^1-h_2+h_1),\phi(S_1,A_1)S_i)\\
&=&\mathrm{dist}_{G}(\phi(S_1,A_1)\pi (h_1),\phi(S_1,A_1)S_i)\\
&=&\mathrm{dist}_{G}(\pi (h_1),S_i)\\
&\geq& \mathrm{dist}_{G}(\pi(j_k),\pi(j_l))-\mathrm{dist}_{G}(\pi(j_k),S_i)-\mathrm{dist}_{G}(\pi(j_l),\pi(t_l)) \\
&\geq& 8(N_0+|S|)+N-N_0-2N_0-|S|\\
&\geq& 5N_0+7|S|+N.
\end{eqnarray*}
For any $u\in\beta_1\cup[\theta_1\setminus(\pi_1(h_2^1-h_2+h_1),\ldots,\pi_1(n-h_2+h_2^1))]\cup\omega_1$, we have
\begin{eqnarray*}
\mathrm{dist}_{G}(u,\pi_1(h_2^1-h_2+h_1))\leq 2|\omega_1|+|\beta_1|\leq 2 N_0+N.
\end{eqnarray*}
Then
\begin{eqnarray*}
\mathrm{dist}_{G}(u,\phi(S_1,A_1)S_i)&\geq&\mathrm{dist}_{G}(\pi_1(h_2^1-h_2+h_1),\phi(S_1,A_1)S_i)-\mathrm{dist}_{G}(u,\pi_1(h_2^1-h_2+h_1)) \\
&\geq& 3N_0+7|S|.
\end{eqnarray*}
For any $u\in \alpha_1\subseteq \pi$, since $\pi$ is an $n$-step SAW on $G_1$, and $G_1$ is a component of $G\setminus \Gamma S$, we deduce that $u\in A_1$ and $\pi(h_i)\in A_1$. Hence
\begin{eqnarray*}
\pi_1(h_i^1)=\phi(S_1,A_1)[\pi(h_i)]\in B_1.
\end{eqnarray*}
Since any path joining $u$ and $\pi_1(h_i^1)$ must cross $S_1$, we obtain
\begin{eqnarray*}
\mathrm{dist}_{G}(u,\phi(S_1,A_1)S_i)&\geq& \mathrm{dist}_{G}(u,\pi_1(h_i^1))-\mathrm{dist}_{G}(\pi_1(h_i^1),\phi(S_1,A_1)S_i)\\
&\geq & \mathrm{dist}_{G}(u,S_1)+\mathrm{dist}_{G}(\pi_1(h_i^1),S_1)-\mathrm{dist}_{G}(\pi_1(h_i^1),\phi(S_1,A_1)S_i)\\
&\geq &2N_0-N_0=N_0.
\end{eqnarray*}
This way we obtain (\ref{c}), and the lemma follows.
\end{proof}

Let 
\begin{eqnarray*}
\pi_0:=\pi,\qquad h_i^0=h_i,
\end{eqnarray*}
for $1\leq i\leq \delta n$.
Assume that we have constructed SAWs $\pi_1,\ldots, \pi_s$, where $1\leq s<\delta n$. Assume that 
$\alpha_s=(\pi_{s-1}(0),\ldots,\pi_{s-1}(h_s^{s-1}))$ is a subwalk of both $\pi_{s-1}$ and $\pi_s$, and that $\pi_{s-1}\setminus \alpha_s$ is mapped by a graph automorphism $\phi_s\in \Gamma$ to a subwalk of $\pi_s$. Let $h_i^s=h_i^{s-1}$, for $1\leq i\leq s$, and let $h_i^s\ (s+1\leq i\leq \delta n)$ be the step in $\pi_s$ of the image of $\pi_{s-1}(h_i^{s-1})$ under the graph automorphism $\phi_s$. Assume also that $\pi_s(h_q^s)$ is the closest vertex, in graph distance, on $\pi_s$ to $\phi_s\ldots\phi_1 S_q$ for $s+1\leq q\leq\delta n$.

 Now we construct an SAW $\pi_{s+1}$, following the procedure below.
 \begin{itemize}
 \item let $\alpha_{s+1}$ be the subwalk of $\pi_s$ from $\pi_s(0)$ to $\pi_s(h_{s+1}^s)$;
 \item use a shortest path $\omega_{s+1}$ to join $\pi_s(h_{s+1}^s)$ and $\tilde{S}_{s+1}:=\phi_s\ldots\phi_1 S_{s+1}$;
 \item let $A_{s+1}$ be the component of $G\setminus \tilde{S}_{s+1}$ containing $\pi_s(h_{s+1}^s)$; use $\phi_{s+1}:=\phi(\tilde{S}_{s+1},A_{s+1})$ to map the concatenation of the reversed $\omega_{s+1}$ and $\pi_s(h_{s+1}^s),\ldots,\pi_s(n+h_{s+1}^s-h_{s+1})$ to another component of $G\setminus \tilde{S}_{s+1}$, and denote the image of the concatenation of the two subwalks under $\phi(\tilde{S}_{s+1},A_{s+1})$ by $\theta_{s+1}$;
 \item use an SAW $\beta_{s+1}$ in $G\setminus[A_{s+1}\cup\phi(\tilde{S}_{s+1},A_{s+1})A_{s+1}]$ to join the last vertex of $\omega_{s+1}$ and its image under $\phi_{s+1}$ (note that the existence of $\beta_{s+1}$ is guaranteed by Assumption \ref{ap31});
 \item let $\pi_{s+1}$ be the concatenation of $\alpha_{s+1}$, $\omega_{s+1}$, $\beta_{s+1}$ and $\theta_{s+1}$.
 \end{itemize}
 We can check that $\pi_{s+1}$ is indeed an SAW. See Figure \ref{fig:reg}.
   \begin{figure}[hb]
  \centering
  \includegraphics[width=15cm]{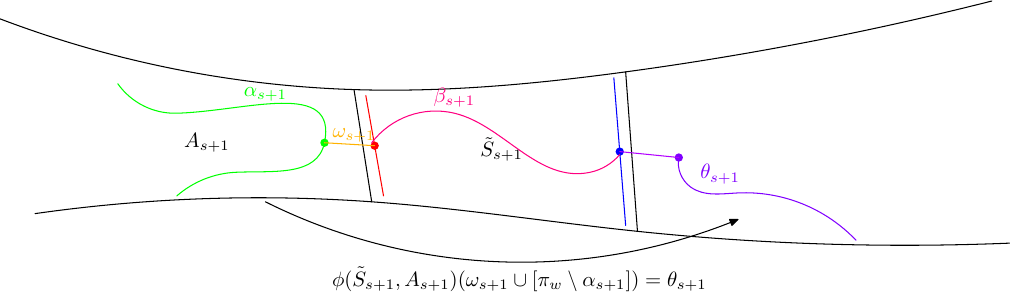}
  \caption{}\label{fig:reg}
\end{figure}

 We can also show that that $\pi_{s+1}(h_p^{s+1})$ is a closest vertex on $\pi_{s+1}$ to $\phi_{s+1}\ldots\phi_1 S_p$, for $s+2\leq p\leq\delta n$, where $\pi_{s+1}(h_p^{s+1})=\phi_{s+1}(\pi_s(h_p^s))$, by similar arguments as in the proof of Lemma \ref{ll25}.
 We repeat the above induction process until we construct the SAW $\pi_{\delta n}$. Note that $\pi_{\delta n}$ is an SAW whose length is at least $n$ and at most $n+\delta n(N+2N_0)$.

 We consider the number of pairs $(\pi,H)$, and obtain
\begin{eqnarray}
|(\pi,H)|\geq \overline{c}_n(v)\left(\begin{array}{c}\kappa n\\\delta n\end{array}\right).\label{p1}
\end{eqnarray}
Moreover,
\begin{eqnarray}
|(\pi, H)|\leq \left(\sum_{t=n}^{n+\delta n(N+2N_0)}c_t(v)\right)\left[d^{|S|}|S|(2N_0+N)^4\right]^{2\delta n}\label{pp2}
\end{eqnarray}
where $d$ is the maximal vertex degree of $G$; $d$ is obviously finite since $G$ is quasi-transitive and locally finite. Taking $n$th root of (\ref{p1}), (\ref{pp2}), and letting $n\rightarrow\infty$, we have
\begin{eqnarray}
\lambda(\tilde{E}_1)\frac{\kappa^{\kappa}}{\delta^{\delta}(\kappa-\delta)^{\kappa-\delta}}\leq \mu\left(d^{|S|}|S|(2N_0+N)^4\right)^{2\delta}.\label{req}
\end{eqnarray}
Then the proof is completed by the following lemma:
\begin{lemma}If (\ref{req}) holds, we have $\lambda(\tilde{E})<\mu$.
\end{lemma}
\begin{proof}
We claim that (\ref{req}) cannot hold when $\lambda(\tilde{E}_1)=\mu$, and $\frac{\delta}{\kappa}>0$ and is sufficiently small. That is because
\begin{eqnarray*}
&&\lim_{\delta\rightarrow 0}\frac{\kappa^{\kappa}}{\delta^{\delta}(\kappa-\delta)^{\kappa-\delta}}=1.\\
&&\lim_{\delta\rightarrow 0}\left(d^{|S|}|S|(2N_0+N)^4\right)^{2\delta }=1.
\end{eqnarray*}
Let
\begin{eqnarray*}
f_1(\delta)&=&\kappa \log \kappa-\delta\log \delta-(\kappa-\delta)\log (\kappa-\delta)\\
f_2(\delta)&=&2\delta\log[d^{|S|}|S|(2N_0+N)^4]
\end{eqnarray*}
Then
\begin{eqnarray*}
f_1'(\delta)=\log \frac{\kappa-\delta}{\delta};\qquad f_2'(\delta)=2 \log[d^{|S|}|S|(2N_0+N)^4]
\end{eqnarray*}
Hence
\begin{eqnarray*}
\lim_{\delta\rightarrow 0}f_1'(\delta)=+\infty;\qquad \lim_{\delta\rightarrow 0}f_2'(\delta)=2 \log[d^{|S|}|S|(2N_0+N)^4]<+\infty.
\end{eqnarray*}
Then when $\delta$ is sufficiently small, we have
\begin{eqnarray*}
\frac{\kappa^{\kappa}}{\delta^{\delta}(\kappa-\delta)^{\kappa-\delta}}> \left(d^{|S|}|S|(2N_0+N)^4\right)^{2\delta};
\end{eqnarray*}
which contradicts (\ref{req}). Therefore we have $\lambda(\tilde{E}_1)<\mu$.
\end{proof}
\end{proof}

To prove Theorem \ref{m31}, we will analyze both the case $\lambda(E^*)=\mu$ and the case $\lambda(E^*)<\mu$.

\begin{lemma}
\begin{enumerate}
\item If 
\begin{eqnarray}
\lambda(E^*)=\mu\label{ls};
\end{eqnarray} 
then there exists an integer $k$ with $1\leq k<|S|$, such that
\begin{enumerate}
\item there exists a vertex $v_0$ such that 
\begin{eqnarray}
\limsup_{n\rightarrow\infty}c_n^{v_0}(0,\tilde{E}_{k+1})^{\frac{1}{n}}=\mu.\label{ct1}
\end{eqnarray}

\item for any vertex $v_0$ of $G$,
\begin{eqnarray}
\limsup_{n\rightarrow\infty}c_n^{v_0}(an, \tilde{E}_k(m))^{\frac{1}{n}}<\mu,\label{ct2}
\end{eqnarray}
where $a,m$ satisfy (\ref{pst}).
\end{enumerate}
\item  If 
\begin{eqnarray}
\lambda(E^*)<\mu\label{lmu};
\end{eqnarray} 
 for any vertex $v_0$ of $G$,
\begin{eqnarray}
\limsup_{n\rightarrow\infty}c_n^{v_0}(an,E^*(m))^{\frac{1}{n}}<\mu.\label{aml}
\end{eqnarray}
\end{enumerate}
\end{lemma}

\begin{proof} We first assume (\ref{ls}). Let $1\leq k\leq |S|$. We make the following observations. First, $c_n(0,\tilde{E}_k)$ is a non-decreasing function of $k$; secondly, if $E^*$ does not occur on a given walk, then $E_{|S|}$ cannot occur. Therefore
\begin{eqnarray*}
c_n(0,E^*)\leq c_n(0,\tilde{E}_{|S|})\leq c_n.
\end{eqnarray*}
As a result, (\ref{ls}) implies
\begin{eqnarray*}
\limsup_{n\rightarrow\infty} c_n(0,\tilde{E}_{|S|})^{\frac{1}{n}}=\mu.
\end{eqnarray*}

By Lemma \ref{l2}, $\lambda(\tilde{E}_1)<\mu$. We may choose $k$ with $1\leq k<|S|$, such that 
\begin{eqnarray}
\lambda(\tilde{E}_k)<\mu,\qquad\lambda(\tilde{E}_{k+1})=\mu\label{dk}
\end{eqnarray}

Let $\epsilon,m$ satisfy (\ref{em1}) and let $m\geq M$ satisfy (\ref{em2}). By Lemma \ref{l1}, there exists $a=a(\epsilon,m)>0$, such that 
\begin{eqnarray}
\limsup_{n\rightarrow\infty} c_n(an, \tilde{E}_k(m))^{\frac{1}{n}}<\mu.\label{pst}
\end{eqnarray}

By quasi-transitivity of $G$, $\lambda(\tilde{E}_{k+1})=\mu$ implies (\ref{ct1}).

Moreover, $\lambda(\tilde{E}_k)<\mu$ implies (\ref{ct2}).

Now assume that (\ref{lmu}), by quasi-transitivity of $G$, (\ref{cc1}) and Lemma \ref{l1}, for any vertex $v_0$ of $G$,
\begin{eqnarray}
\limsup_{n\rightarrow\infty}c_n^{v_0}(an,E^*(m))^{\frac{1}{n}}<\lim_{n\rightarrow\infty}c_n(v_0)^{\frac{1}{n}}=\mu.\label{aml}
\end{eqnarray}
\end{proof}

Let $\pi$ be an $n$-step SAW on $G$. We say that $F$ occurs at the $j$th step of $\pi$ when one of the following two cases occur:
\begin{letlist}
\item if $\lambda(E^*)<\mu$, then
\begin{itemize}
\item $E^*(m)$ occurs at the $j$th step of $\pi$; and
\item assume that $\pi(j)\in \gamma S$, $\gamma\in \Gamma$ and the subwalk $(\pi(j-m),\ldots,\pi(j+m))$ visits all the vertices of $\gamma S$. Let $\pi(\alpha)$ ($\alpha\geq j-m$) be the first vertex of $\gamma S$ visited by $\pi$, and let $\pi(\beta)$ ($\alpha<\beta\leq j+m$) be the last vertex of $\gamma S$ visited by $\pi$. Then $\pi(\alpha-1)$ and $\pi(\beta+1)$ are in distinct components of $G\setminus \gamma S$.
\end{itemize}
\item if $\lambda(E^*)=\mu$, let $k$ be as in (\ref{dk}), then
\begin{itemize}
\item $\tilde{E}_k(m)$ occurs at the $j$th step, and $\tilde{E}_{k+1}$ does not occur at the $j$th step; and
\item assume that $\pi(j)\in \gamma S$, $\gamma\in \Gamma$ and $\pi$ visits exactly $k$ vertices of $\gamma S$. Let $\pi(\alpha)$ ($\alpha\geq j-m$) be the first vertex of $\gamma S$ visited by $\pi$, and let $\pi(\beta)$ ($\alpha<\beta\leq j+m$) be the last vertex of $\gamma S$ visited by $\pi$. Then $\pi(\alpha-1)$ and $\pi(\beta+1)$ are in distinct components of $G\setminus \gamma S$.
\end{itemize}
\end{letlist}

For $r\geq 0$, and $v\in V$
\begin{itemize}
\item if $\lambda(E^*)<\mu$, let $b_n^v(r,F)$ be the number of $n$-step SAWs on $G$ starting from $v$, such that $E^*(m)$ occurs at least $an$ times, and $F$ occurs no more than $r$ steps;
\item if $\lambda(E^*)=\mu$, let $b_n^v(r,F)$ be the number of $n$-step SAWs on $G$ starting from $v$, such that $\tilde{E}_k(m)$ occurs at least $an$ times, $E_{k+1}$ never occurs, and $F$ occurs in no more than $r$ steps. 
\end{itemize}

\begin{lemma}\label{l25}
\begin{letlist}
\item If $\lambda(E^*)<\mu$, then for any $v\in V$, we have
\begin{eqnarray*}
\limsup_{n\rightarrow\infty}b_n^v(0,F)^{\frac{1}{n}}<\mu.
\end{eqnarray*}
\item If $\lambda(E^*)=\mu$, let $v$ be a vertex satisfying (\ref{ct1}), then
\begin{eqnarray*}
\limsup_{n\rightarrow\infty}b_n^v(0,F)^{\frac{1}{n}}<\mu.
\end{eqnarray*}
\end{letlist}
\end{lemma}
\begin{proof}Let $v$ be a vertex satisfying the assumptions of the lemma. Assume that
\begin{eqnarray}
\limsup_{n\rightarrow\infty}b_n^v(0,F)^{\frac{1}{n}}=\mu;\label{bem}
\end{eqnarray}
we will obtain a contradiction. The idea is to modify those SAWs where $F$ never occurs to SAWs where $F$ occurs $\delta N (0<\delta<1)$ times; different ways of modifications give a nontrivial exponential factor to the total number of $n$-step SAWs compared to $b_n^v(0,F)$. Therefore if (\ref{bem}) holds, then the connective constant must be strictly greater than $\mu$.

We define a set $U_n$ of SAWs as follows
\begin{itemize}
\item If $\lambda(E^*)<\mu$, let $U_n$ be the set consisting of all the $n$-step SAWs on $G$ starting from $v$ such that $E^*(m)$ occurs at least $an$ times, and $F$ never occurs. .
\item If $\lambda(E^*)=\mu$, let $k$ be given as in (\ref{dk}) and let $U_n$ be the set consisting of all the $n$-step SAWs on $G$ starting from $v$, such that $\tilde{E}_k(m)$ occurs at least $an$ times, and $\tilde{E}_{k+1}$ never occurs, and $F$ never occurs.
\end{itemize}

Let $\pi\in U_n$. Let $j_1,\ldots,j_u$ be the indices of the SAW $\pi$ such that one of the following is true 
\begin{itemize}
\item $E^*(m)$ occurs at $j_1,\ldots,j_u$ if $\lambda(E^*)<\mu$; or
\item $E_k(m)$ occurs if $\lambda(E^*)=\mu$. 
\end{itemize}
Moreover, assume that for any $1\leq i<\ell\leq u$,
\begin{eqnarray}
\mathrm{dist}_G(\pi(j_i),\pi(j_{\ell}))\geq 4(2m+1)|S|.\label{dil}
\end{eqnarray}

We may assume $u=\kappa n$ such that (\ref{dil}) holds. We choose a subset 
\begin{eqnarray*}
H=\{t_1,\ldots, t_{\delta n}\}\subset\{j_1,\ldots, j_{\kappa n}\},
\end{eqnarray*}
where $0<\delta<\kappa$, $t_i<t_j$ when $i<j$, to perform the following inductive manipulations.

 For $1\leq i\leq \delta n$, let $\gamma_i S$ be the copy of $S$ (where $\gamma_i\in\Gamma$) such that $\pi(t_i)\in \gamma_i S$, and all the vertices of $\gamma_i S$ are visited by $\pi$ and by the subwalk $(\pi(t_i-m),\ldots,\pi(t_i+m))$ if $\lambda(E^*)<\mu$ (exactly $k$ vertices of $\gamma_i S$ are visited by the subwalk $(\pi(t_i-m),\ldots,\pi(t_i+m))$ if $\lambda(E^*)=\mu$). Let $\pi(\alpha_i)$ ($\alpha_i\geq t_i-m$) be the first vertex of $\gamma_i S$ visited by $\pi$, and let $\pi(\beta_i)$ ($\alpha_i<\beta_i\leq t_i+m$) be the last vertex of $\gamma_i S$ visited by $\pi$. Since $\pi\in U_n$, $F$ never occurs in $\pi$; therefore $\pi(\alpha_i-1)$ and $\pi(\beta_i+1)$ are in the same component of $G\setminus \gamma_i S$.

Let $A_i$ be the component of $G\setminus \gamma_i S$ containing $\pi(\alpha_i-1)$ and $\pi(\beta_i+1)$, and let $\phi(\gamma_i S,A_i)$ be as described in Assumption \ref{ap31} (3) and $\phi(\gamma_i S,A_i)A_i \subseteq B_i$. We construct a new SAW $\pi_1$ from $\pi$ as follows.
\begin{itemize}
\item The subwalk $(\pi_1(0),\ldots,\pi_1(\alpha_1))$ is the same as the subwalk $(\pi(0),\ldots,\pi(\alpha_1))$.
\item We  map $(\pi(\beta_1+1),\ldots, \pi(n))$, as an SAW in $A_1$, to an SAW $\theta_1$ in $B_1$ by $\phi(\gamma_1 S,A_1)$, as described in Assumption \ref{ap31} (3). 
\item we use an SAW $\omega_1$ in $G\setminus[A_1\cup\phi(\gamma_1S,A_1)A_1]$ joining $\pi(\alpha_1)$ and the vertex $\phi(\gamma_1 S,A_1)\pi(\beta_1)$.  This is possible since $S$ is connected by Assumption \ref{ap31} (1), and $\gamma_1 S$ is an identical copy of $S$.
\item Let $\pi_1$ be the concatenation of $(\pi_1(0),\ldots,\pi_1(\alpha_1))$, $\omega_1$ and $\theta_1$.
\end{itemize}
See Figure \ref{fig:rell2}.

  \begin{figure}[hb]
  \centering
  \includegraphics[width=15cm]{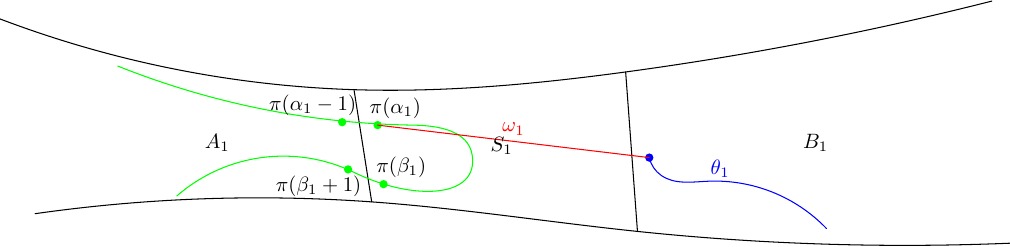}
  \caption{The self-avoiding walk $\pi$ is represented by the green curve; $\theta_1=\phi(\gamma_1 S,A_1)(\pi(\beta_1+1),\ldots,\pi(n))$; $\pi_1=(\pi(0),\ldots,\pi(\alpha_1))\cup\omega_1\cup\theta_1$}\label{fig:rell2}
\end{figure}

Let 
\begin{eqnarray*}
\pi_0:=\pi,\qquad \alpha_i^0:=\alpha_i,\qquad \beta_i^0:=\beta_i,\qquad \phi_1:=\phi(\gamma_1 S,A_1).
\end{eqnarray*}
for $1\leq i\leq \delta n$.
We make the following induction hypothesis:
assume that we have constructed SAWs $\pi_1,\ldots, \pi_k$,  and graph automorphisms $\phi_1,\ldots,\phi_k$, where $1\leq k<\delta n$ such that
\begin{itemize}
\item  for $1\leq i\leq k$, 
$\eta_i=(\pi_{i-1}(0),\ldots,\pi_{i-1}(\alpha_i^{i-1}))$ is a subwalk of both $\pi_{i-1}$ and $\pi_i$; and 
\item for $1\leq i\leq k$, $\tilde{S}_i=\phi_{i-1}\ldots \phi_1 S_i$; and
 \item \begin{eqnarray*}
&&\alpha_i^l=\alpha_i^{l-1}, \qquad\mathrm{for}\ 1\leq i\leq l\leq k-1;\\
&&\beta_j^l=\beta_j^{l-1}, \qquad\mathrm{for}\ 1\leq j<l\leq k-1;
\end{eqnarray*}
and for $1\leq l\leq k-1$, $\alpha_i^l\ (l+1\leq i\leq \delta n)$ (resp. $\beta_j^l\ (l\leq j\leq \delta n)$) is the step in $\pi_l$ of the image of $\pi_{l-1}(\alpha_i^{l-1})$ (resp. $\pi_{l-1}(\beta_j^{l-1})$) under the graph isomorphism $\phi_l$; and

\item for $1\leq i\leq k$, $(\pi_{i-1}(\beta_i^{i-1}),\pi_{i-1}(\beta_i^{i-1}+1),\ldots,\pi_{i-1}(|\pi_{i-1}|))$ is mapped by the graph automorphism $\phi_k\in \Gamma$ to a subwalk of $\pi_i$.
\end{itemize}
 Let 
\begin{eqnarray*}
&&\alpha_i^k=\alpha_i^{k-1}, \qquad\mathrm{for}\ 1\leq i\leq k;\\
&&\beta_j^k=\beta_j^{k-1}, \qquad\mathrm{for}\ 1\leq j\leq k-1.
\end{eqnarray*}
Let $\alpha_i^k\ (k+1\leq i\leq \delta n)$ (resp. $\beta_j^k\ (k\leq j\leq \delta n)$) be the step in $\pi_k$ of the image of $\pi_{k-1}(\alpha_i^{k-1})$ (resp. $\pi_{k-1}(\beta_j^{k-1})$) under the graph isomorphism $\phi_k$.  . 

 Now we construct an SAW $\pi_{k+1}$, following the procedure below.
 \begin{itemize}
 \item Let $\eta_{k+1}$ be the subwalk of $\pi_k$ from $\pi_k(0)$ to $\pi_k(\alpha_{k+1}^k)$.
\item Let $\tilde{S}_{k+1}:=\phi_k\ldots\phi_1 S_{k+1}$.
 Let $A_{k+1}$ be the component of $G\setminus \tilde{S}_{k+1}$ containing $\pi_k(\alpha_{k+1}^k)$; use $\phi_{k+1}:=\phi(\tilde{S}_{k+1},A_{k+1})$ to map the subwalk $\pi_k(\beta_{k+1}^k),\ldots,\pi_k(n-\beta_{k+1}+\beta_{k+1}^k)$ to an SAW $\theta_{k+1}$ in  another component of $G\setminus \tilde{S}_{k+1}$.
 \item Use an SAW $\omega_{k+1}$ in $G\setminus[A_{k+1}\cup\phi_{k+1}(\tilde{S}_{k+1},A_{k+1})A_{k+1}]$ to join $\pi_k(\alpha_{k+1}^k)$ and $\phi_{k+1}(\pi_k(\beta_{k+1}^k))$ (note that the existence of such an SAW is guaranteed by Assumption \ref{ap31});
 \item let $\pi_{k+1}$ be the concatenation of $\eta_{k+1}$, $\omega_{k+1}$ and $\theta_{k+1}$.
 \end{itemize}
 We can check that $\pi_{k+1}$ is indeed an SAW. See Figure \ref{fig:rell3}.
 
 \begin{figure}[hb]
  \centering
  \includegraphics[width=15cm]{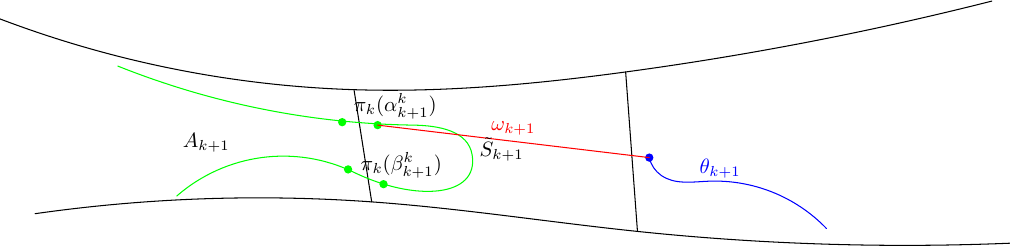}
  \caption{The self-avoiding walk $\pi_k$ is represented by the green curve; $\theta_{k+1}=\phi(\tilde{S}_{k+1},A_{k+1})(\pi_k(\beta_{k+1}^k+1),\ldots,\pi_k(|\pi_k|))$; $\pi_{k+1}=(\pi_k(0),\ldots,\pi_k(\alpha_{k+1}^k))\cup\omega_{k+1}\cup\theta_{k+1}$}\label{fig:rell3}
\end{figure}
 
 We claim that 
 \begin{eqnarray}
\left[\phi_k\ldots\phi_1 S_{k+1}\right]\cap\left[\phi_{k+1}\ldots\phi_1 S_{k+2}\right]=\emptyset.\label{ie}
 \end{eqnarray}
 To see why (\ref{ie}) is true, note that by (\ref{dil}), we have
 \begin{eqnarray*}
 \mathrm{dist}_G(S_{k+1},S_{k+2})\geq 2|S|.
 \end{eqnarray*}
 Since $\phi_i\ (1\leq i\leq k)$'s are automorphisms of $G$, we have
 \begin{eqnarray*}
  \mathrm{dist}_G(\phi_k\ldots\phi_1S_{k+1},\phi_k\ldots\phi_1 S_{k+2})\geq 2|S|.
 \end{eqnarray*}
 This implies 
  \begin{eqnarray*}
[\phi_k\ldots\phi_1S_{k+1}]\cap[\phi_k\ldots\phi_1 S_{k+2}]=\emptyset.
 \end{eqnarray*}
 Since $\phi_{k+1}$ maps $\phi_k\ldots\phi_1 S_{k+2}$ from one component of $G\setminus[\phi_k\ldots\phi_1S_{k+1}]$ to another component of of $G\setminus[\phi_k\ldots\phi_1S_{k+1}]$, (\ref{ie}) follows.
 
 We continue the above construction process until we have already constructed the SAW $\pi_{\delta n}$. Note that the length of $\pi_{\delta n}$ satisfies
 \begin{eqnarray*}
 |\pi_{\delta n}|\leq n+n\delta N;
 \end{eqnarray*}
 since the length of $\omega_i (1\leq i\leq \delta n)$ satisfies $|\omega_i|\leq N$ by Assumption \ref{ap31}(3).
 
  Counting the number of pairs $|(\pi,H)|$, we have
\begin{eqnarray*}
|(\pi,H)|\geq |U_n|\left(\begin{array}{c}\kappa n\\\delta n\end{array}\right),
\end{eqnarray*}
and
\begin{eqnarray*}
|(\pi,H)|\leq \left(\sum_{i=1}^{n+\delta nN} c_i\right)\left(\sum_{i=1}^{2m} c_i d^{|S|}|S|(N+d^{|S|})^3\right)^{\delta n}.
\end{eqnarray*}
We have 
\begin{eqnarray*}
\limsup_{n\rightarrow\infty}|U_n|^{\frac{1}{n}}\leq \frac{\delta^{\delta}(\kappa-\delta)^{\kappa-\delta}}{\kappa^{\kappa}}\mu^{1+2N\delta}\left[\sum_{i=1}^{2m} c_i d^{|S|}|S|(N+d^{|S|})^3\right]^{\delta}:=g(\delta)
\end{eqnarray*}
Note that 
\begin{eqnarray*}
g(0)=\mu,\qquad g'(0)=-\infty.
\end{eqnarray*}
Hence for $0<\delta\leq\delta_0$, $g(\delta)<\mu$. Therefore we have 
\begin{eqnarray*}
\limsup_{n\rightarrow\infty}|U_n|^{\frac{1}{n}}<\mu,
\end{eqnarray*}
and the proof is complete.
\end{proof}

\begin{lemma}\label{l29}For $r\geq 0$, let
\begin{eqnarray*}
b_n(r,F)=\sup_{v\in G}b_n^v(r,F).
\end{eqnarray*}
Then
\begin{eqnarray}
\limsup_{n\rightarrow\infty}b_n(0,F)^{\frac{1}{n}}<\mu.\label{lm}
\end{eqnarray}
\end{lemma}

\begin{proof}If $\lambda(E^*)<\mu$, then (\ref{lm}) follows from Part (a) of Lemma \ref{l25}, and the quasi-transitivity of $G$ under the action of $\Gamma$.

Now assume that $\lambda(E^*)=\mu$. Let $k$ be a positive integer satisfying (\ref{dk}). We consider the following cases
\begin{Alist}
\item The starting vertex $v$ of SAWs satisfy
\begin{eqnarray*}
\limsup_{n\rightarrow\infty}c_n^v(0,\tilde{E}_{k+1})^{\frac{1}{n}}=\mu.
\end{eqnarray*}
\item The starting vertex $v$ of SAWs satisfy
\begin{eqnarray}
\limsup_{n\rightarrow\infty}c_n^v(0,\tilde{E}_{k+1})^{\frac{1}{n}}<\mu. \label{clm}
\end{eqnarray}
\end{Alist}
In Case A., (\ref{lm}) follows from Part (b) of Lemma \ref{l25}. In Case B., (\ref{lm}) follows from (\ref{clm}) and the following fact 
\begin{eqnarray*}
b_n^v(0,F)\leq c_n^v(0,\tilde{E}_{k+1}).
\end{eqnarray*}
\end{proof}

\begin{lemma}\label{lf}There exist $a',R'\in(0,1)$, such that
\begin{eqnarray*}
\limsup_{n\rightarrow\infty} b_n(a'n,F)^{\frac{1}{n}}<R'\mu.
\end{eqnarray*}
\end{lemma}
\begin{proof}Let us consider the number of $n$-step SAWs on $G$ starting from a fixed vertex, such that $F$ occurs less than $\xi n$ times, where $\xi$ is a small positive number.

Let $\epsilon,m$ satisfy (\ref{em2}). 
If $F$ occurs in less than $\xi n$ times, then it occurs in less than $\xi n$ of the $L=\left\lfloor\frac{n}{m}\right\rfloor$ $m$-step subwalks 
\begin{eqnarray*}
(\pi(im+1),\ldots, \pi((i+1)m)),\qquad \mathrm{for}\ 0\leq i\leq n.
\end{eqnarray*}

By enumerating all the $m$-step subwalks where $F$ occurs, we obtain
\begin{eqnarray}
b_n(\xi n, F)\leq \sum_{i=1}^{\xi n}\left(\begin{array}{c}L\\ i \end{array}\right)[\mu(1-\epsilon)]^{Nm-im}(\mu(1+\epsilon))^{im}c_{n-Nm}.\label{bnl}
\end{eqnarray}

Note that by Lemma \ref{l29}, we have
\begin{itemize}
\item When $\lambda(E^*)<\mu$, $[\mu(1-\epsilon)]^m$ gives an upper bound for the number of $m$-step SAWs starting from a fixed vertex, such that one of the followings holds
\begin{itemize}
 \item $E^*(m)$ occurs less than $am$ times; or
 \item $E^*(m)$ occurs at least $am$ times, and $F$ never occurs,
 \end{itemize}
 where $a,m$ satisfies (\ref{aml}).
 \item When $\lambda(E^*)=\mu$, $[\mu(1-\epsilon)]^m$ gives an upper bound for the number of $m$-step SAWs starting from a fixed vertex such that 
\begin{itemize}
\item $\tilde{E}_k(m)$ occurs less than $am$ times; or
\item $\tilde{E}_k(m)$ occurs at least $am$ times, $E_{k+1}$ never occurs, and $F$ never occurs.
\end{itemize}
Here $k$ satisfies (\ref{dk}), and $a,m$ satisfies (\ref{pst}).
\end{itemize}
Taking $n$th roots of \ref{bnl} and letting $n\rightarrow \infty$, we obtain that when $\xi$ is sufficiently small
\begin{eqnarray*}
\limsup_{n\rightarrow\infty}b_n(\xi n,F)^{\frac{1}{n}}<\mu,
\end{eqnarray*}
and the lemma follows.
\end{proof}

\bigskip
\noindent\textbf{Proof of Theorem \ref{m31} A.}
If $\lambda(E^*)<\mu$, let $d_n^v$ be the number of $n$-step SAWs in $G$ starting from $v$ such that $E^*(m)$ occurs at least $an$ times, and $F$ occurs at least $a'n$ times, where $a$ is given by (\ref{aml}) and $a'$ is given by Lemma \ref{lf}. Then we have
\begin{eqnarray}
\limsup_{n\rightarrow\infty}d_n^v=\mu.\label{pt1}
\end{eqnarray}

Now assume that $\lambda(E^*)=\mu$. Let $v$ be a vertex of $G$ satisfying (\ref{ct1}). Let $a>0$ be given by (\ref{ct2}), and $a'>0$ be given by Lemma \ref{lf}. Let $d_n^{v}$ be the number of $n$-step SAWs in $G$ such that $E_{k+1}$ never occurs, $E_k(m)$ occurs at least $an$ times, and $F$ occurs at least $a'n$ times. Then
\begin{eqnarray}
\limsup_{n\rightarrow\infty}d_n^{v}=\mu.\label{pt2}
\end{eqnarray}

Note that in an SAW $\pi$, if $F$ occurs at least $a'n$ times, then $\|\pi\|\geq a'n$. The theorem follows from (\ref{pt1}) and (\ref{pt2}). $\hfill\Box$

\section{Proof of Theorem \ref{m31} B.}\label{ppb}

We prove Theorem \ref{m31} B. in this section. Let $G=(V,E)$ be a graph satisfying Assumption \ref{ap32}. Since any graph satisfying Assumption \ref{ap32} must satisfy Assumption \ref{ap31} as well, all the results proved in Section \ref{pp1} also apply to graphs satisfying Assumption \ref{ap32}.

Let $\pi$ be an $n$-step SAW on $G$. Recall that $E^*$ occurs at the $j$th step of $\pi$ if there exists $\gamma S\in \Gamma S$ such that $\pi(j)\in \gamma S$, and all the vertices of $\gamma S$ are visited by $\pi$. For $k\geq 1$, we say that $\mathcal{E}_k$ occurs at the $j$th step of $\pi$, if there exists $\gamma S'\in \Gamma S'$, such that $\pi(j)\in \gamma S'$, and at least $k$ vertices of $\gamma S'$ are visited by $\pi$, where $S'$ containing $S$ is given as in Assumption \ref{ap32}. We say that $\tilde{\mathcal{E}}_k$ occurs at the $j$th step of $\pi$ if $E^*$ or $\mathcal{E}_k$ (or both) occur there.

In the following, we will use $\mathcal{E}$ to denote any of $E^*$, $\mathcal{E}_k$ or $\tilde{\mathcal{E}}_k$. If $m$ is a positive integer, we say that $\mathcal{E}(m)$ occurs at the $j$th step of $\pi$ if $\mathcal{E}$ occurs at the $m$th step of the $2m$-step subwalk $(\pi(j-m),\ldots,\pi(j+m))$. 

For $r\geq 0$ and $v\in G$, let $c_n^v(r,\mathcal{E})$ (resp.\ $c^v_n(r,\mathcal{E}(m))$) be the number of $n$-step SAWs starting from $v$ for which $\mathcal{E}$ (resp.\ $\mathcal{E}(m)$) occurs at no more than $r$ different  steps. Let
\begin{eqnarray*}
c_n(r,\mathcal{E})&=&\sup_{v\in V}c_n^v(r,\mathcal{E})\\
c_n(r,\mathcal{E}(m))&=&\sup_{v\in V}c^v_n(r,\mathcal{E}(m))
\end{eqnarray*}

Let 
\begin{eqnarray*}
\lambda(\mathcal{E})=\limsup_{n\rightarrow\infty}c_n(0,\mathcal{E})^{\frac{1}{n}}.
\end{eqnarray*}
Let $\mu$ be the connective constant of $G$. We have that $\lambda(\mathcal{E})<\mu$ if and only if 
\begin{eqnarray}
\mathrm{there\ exist\ }\epsilon>0, M\in\NN, \mathrm{such\ that\ } c_m(0,\mathcal{E})<[\mu(1-\epsilon)]^m, \mathrm{for\ } m\geq M.\label{eem1}
\end{eqnarray}

The proof of Theorem \ref{m31} B. is similar to that of Theorem \ref{m31} A.; both of which are inspired by the proof of the pattern theorem by Kesten (\cite{HK63}), and obtained by modifying local configurations of SAWs on $\delta n (0<\delta <1)$ disjoint locations to create a nontrivial exponential factor on the total number of SAWs. The difference lies in that under the stronger Assumption \ref{ap32}, we obtain the stronger Lemma \ref{l33} (recall that if we only have Assumption \ref{ap31}, we always discuss two cases $\lambda(E^*)<\mu$ and $\lambda(E^*)=\mu$). With Lemma \ref{l33}, we can prove the stronger result that SAWs on such a graph actually have positive speed. More precisely, not only that the exponential growth rate of the number of $n$-step SAWs whose end-to-end distance is linear in $n$ is the same as the connective constant; but the number of $n$-step SAWs whose end-to-end distance is not linear in $n$ is indeed exponentially small compared to the number of all the $n$-step SAWs.

\begin{lemma}\label{ll1}Let $k$ satisfy $1\leq k\leq |S|$, and 
\begin{eqnarray}
\lambda(\mathcal{E})<\mu.\label{llk}
\end{eqnarray}
Let $\epsilon,M$ satisfy (\ref{eem1}), and let $m\geq M$ satisfy
\begin{eqnarray}
c_m\leq [\mu(1+\epsilon)]^m.\label{eem2}
\end{eqnarray}
Then there exist $a=a(\epsilon,m)$ and $R=R(\epsilon,m)\in (0,1)$, such that
\begin{eqnarray*}
\limsup_{n\rightarrow\infty}c_n(an,\mathcal{E}(m))^{\frac{1}{n}}<R\mu.
\end{eqnarray*}
\end{lemma}
\begin{proof}
Same as the proof of Lemma \ref{l1}.
\end{proof}

\begin{lemma}\label{ll2}$\lambda(\tilde{\mathcal{E}}_1)<\mu$.
\end{lemma}
\begin{proof}
The lemma follows from Lemma \ref{l2} and the fact that $\lambda(\tilde{\mathcal{E}}_1)<\lambda(\tilde{E}_1)$.
\end{proof}

\begin{lemma}\label{l33}$\lambda(E^*)<\mu$.
\end{lemma}
\begin{proof} Assume that 
\begin{eqnarray}
\lambda(E^*)=\mu;\label{l34}
\end{eqnarray} 
we will obtain a contradiction. 

We make the following observations. First, $c_n(0,\tilde{\mathcal{E}}_k)$ is a non-decreasing function of $k$; secondly, if $E^*$ does not occur on a given walk, then $\mathcal{E}_{|S'|}$ cannot occur. Therefore
\begin{eqnarray*}
c_n(0,\mathcal{E}^*)\leq c_n(0,\tilde{\mathcal{E}}_{|S'|})\leq c_n.
\end{eqnarray*}
As a result, (\ref{l34}) implies
\begin{eqnarray*}
\limsup_{n\rightarrow\infty} c_n(0,\tilde{\mathcal{E}}_{|S'|})^{\frac{1}{n}}=\mu.
\end{eqnarray*}

By Lemma \ref{ll2}, $\lambda(\tilde{\mathcal{E}}_1)<\mu$. We may choose $k$ with $1\leq k<|S|$, such that 
\begin{eqnarray}
\lambda(\tilde{\mathcal{E}}_k)<\mu,\qquad\lambda(\tilde{\mathcal{E}}_{k+1})=\mu.\label{me1}
\end{eqnarray}

Let $\epsilon,m$ satisfy (\ref{eem1}) and let $m\geq M$ satisfy (\ref{eem2}). By Lemma \ref{ll1}, there exists $a=a(\epsilon,m)>0$, such that 
\begin{eqnarray}
\limsup_{n\rightarrow\infty} c_n(an, \tilde{\mathcal{E}}_k(m))^{\frac{1}{n}}<\mu.\label{me2}
\end{eqnarray}

By (\ref{me1}) and (\ref{me2}) and the quasi-transitivity of $G$, there exists a vertex $v_0$ of $G$, such that 
\begin{eqnarray}
\limsup_{n\rightarrow\infty} c_n^{v_0}(0, \tilde{\mathcal{E}}_{k+1})^{\frac{1}{n}}=\mu;\label{35}
\end{eqnarray}
and for any $v\in V$, 
\begin{eqnarray}
\limsup_{n\rightarrow\infty} c_n^v(an, \tilde{\mathcal{E}}_k(m))^{\frac{1}{n}}<\mu.\label{36}
\end{eqnarray}

Let $T_n$ be the set consisting of all the $n$-step SAWs starting from $v_0$ for which $\mathcal{E}_{k+1}  $ never occurs, but $\tilde{\mathcal{E}}_k(m)$ occurs at least $an$ times.
 We have that
\begin{eqnarray*}
|T_n|\geq c_n^{v_0}(0,\tilde{\mathcal{E}}_{k+1})-c_n^{v_0}(an, \tilde{E}_k(m)).
\end{eqnarray*}

By (\ref{35}) and (\ref{36}), we have 
\begin{eqnarray}
\limsup_{n\rightarrow\infty}|T_n|^{\frac{1}{n}}=\mu.\label{tn}
\end{eqnarray}

 The rest of the proof is devoted to show the existence of $t<1$, such that
\begin{eqnarray*}
\limsup_{n\rightarrow\infty}|T_n|^{\frac{1}{n}}<t\mu.
\end{eqnarray*} 
This contradicts (\ref{tn}), and the lemma follows.

Let $d$ be the maximal vertex degree of $G$. Let $\pi\in T_n$, so that $\pi$ contains at least $an$ occurrences of $\tilde{\mathcal{E}}_k(m)$. We can find $j_1<\ldots< j_u$ with $u=\lfloor\kappa n\rfloor-2$, where 
\begin{eqnarray}
\kappa=\frac{a}{(2m+2)d^{4|S'|}},\label{ka}
\end{eqnarray}
 such that
\begin{eqnarray*}
\tilde{\mathcal{E}}_k(m)\ \mathrm{occurs\ at\ the\ }j_1th, j_2th,\ldots, j_uth\ \mathrm{steps\ of\ }\pi,
\end{eqnarray*}
(and perhaps other steps as well), and in addition,
\begin{itemize}
\item $0<j_1-m$, $j_u+m<n$,
\item $j_t+m<j_{t+1}-m$;
\item for any $\gamma_1,\gamma_2\in\Gamma$, such that $\pi(j_{i})\in \gamma_1S'$ and $\pi(j_{\ell})\in \gamma_2S'$, $i\neq\ell$, $\gamma_1 S'$ and $\gamma_2 S'$ are disjoint.
\end{itemize}

Such $j_t$'s may be found by the following iterative construction. First $j_1$ is the smallest $j>m$ such that $\mathcal{E}_k(m)$ occurs at the $j$th step of $\pi$. Having found $j_1,j_2,\ldots,j_r$, let $j_{r+1}$ be the smallest $j$ such that 
\begin{letlist}
\item $j_{r}+m<j_{r+1}-m$;
\item $\tilde{E}_k(m)$ occurs at the $j$th step of $\pi$;
\item for any $\gamma_1,\gamma_2\in\Gamma$, such that $\pi(j_i)\in \gamma_1 S'$ ($i\leq r$), and $\pi(j_{r+1})\in\gamma_2 S'$, we have $\gamma_1 S'\cap\gamma_2 S'=\emptyset$.
\end{letlist}
Condition (a) gives rise to the factor $(2m+2)$ in the denominator of (\ref{ka}), and Condition (c) gives rise to the factor $d^{4|S'|}$. 

Let $t\in\{1,2,\ldots,u\}$. Since $\tilde{\mathcal{E}}_k(m)$ but not $\tilde{\mathcal{E}}_{k+1}$ occurs at the $j_t$th step, $\pi$ visits at most $k$ vertices in each $\gamma S$ containing $\pi(j_t)$. Let $\Psi_t$ be the set consisting of all the copies $\gamma S'$ of $S'$ in $\Gamma S'$, such that $\pi(j_t)\in \gamma S'$, $\pi$ visits exactly $k$ vertices of $\gamma S'$, and these $k$ vertices lie between $\pi(j_t-m)$ and $\pi(j_t+m)$ on $\pi$. Choose $\gamma_t S\in \Psi_t$. For $t=1,2,\ldots,u$, let 
\begin{eqnarray*}
\alpha_t=\min\{i:\pi(i)\in \gamma_t S'\},\qquad\beta_t=\max\{i:\pi(i)\in \gamma_t S'\},
\end{eqnarray*}
so that
\begin{eqnarray*}
j_t-m\leq \alpha_t\leq j_t\leq \beta_t\leq j_t+m.
\end{eqnarray*}

We next describe the strategy for the replacement of the subwalk $(\pi(\alpha_t),\pi(\alpha_t+1),\ldots,\pi(\beta_t))$. Starting from $\pi(\alpha_t)$, the walk follows an SAW in $\gamma_t S$ joining $\pi(\alpha_t)$ and $\pi(\beta_t)$, and visits every vertex in $\gamma_t S$. Such an SAW exists by Assumption \ref{ap32}.

Let $\delta$ satisfy $0<\delta<\kappa$, to be chosen later, and set $s=\delta n$. Let $H=(h_1,\ldots, h_s)$ be an oriented subset of $\{j_1,\ldots,j_u\}$. We shall make an appropriate substitution in the neighborhood of each $\pi(h_t)$ to obtain an SAW $\pi^*=(\pi,H)$.

We estimate the number of pairs $(\pi,H)$ as follows. First, the number $|(\pi,H)|$ is at least the cardinality of $T_n$ multiplied by the minimum number of possible choices of $H$ as $\pi$ ranges over $T_n$. Any subset of $\{j_1,\ldots, j_u\}$ with cardinality $s=\delta n$ may be chosen for $H$, whence 
\begin{eqnarray}
|(\pi,H)|\geq |T_n|\left(\begin{array}{c}\kappa n-2\\ \delta n\end{array}\right)\label{lb}
\end{eqnarray}
We bound $(\pi, H)$ above by counting the number of SAWs $\pi_*$ of $G$ with length not exceeding $n+|S'|\delta n$, and multiplying by an upper bound for the number of pairs $(\pi,H)$ giving rise to a particular $\pi_*$.

The number of possible choices of $\pi_*$ is no greater than $\sum_{i=0}^{n+|S'|\delta n} c_i$. A given $\pi_*$ contains $|H|=\delta n$ occurrences of visits to all the vertices of some copy of $S'$ under $\Gamma$. At the $t$th such occurrence, $\pi(h_t)$ is a point on $\gamma_t S'$ and there are no more than $|S'|$ different choices of $\pi(h_t)$. For given $\pi_*$ and $(\pi(h_t): t=1,2,\ldots,s)$, there are at most $[\sum_{i=1}^{2m}c_i]^{\delta n}$ corresponding SAWs $\pi$ of $G$. Therefore
\begin{eqnarray}
|(\pi,H)|\leq \left(d^{|S'|}|S'|^2\sum_{i=1}^{2m}c_i\right)^{\delta n}\left(\sum_{i=0}^{n+|S'|\delta n}c_i\right)\label{ub}
\end{eqnarray}

 We combine (\ref{lb}) and (\ref{ub}), take $n$th root and the limit as $n\rightarrow\infty$, we obtain, by the fact that $c_n^{\frac{1}{n}}\rightarrow\mu$,
\begin{eqnarray*}
\mu\frac{\kappa^{\kappa}}{\delta^{\delta}(\kappa-\delta)^{\kappa-\delta}}\leq \left(d^{|S'|}|S'|^2\sum_{i=1}^{2m}c_i\right)^{\delta}\mu^{1+|S'|\delta}.
\end{eqnarray*}
There exists $Z=Z(\epsilon,m,d,|S|)$, such that
\begin{eqnarray*}
d^{|S'|}|S'|^2\mu^{|S'|}\sum_{i=1}^{2m}c_i\leq Z,
\end{eqnarray*}
therefore
\begin{eqnarray*}
\mu\leq f(\eta)^{\kappa}\mu,
\end{eqnarray*}
where $f(\eta)=Z^{\eta}\eta^{\eta}(1-\eta)^{1-\eta}$, and $\eta=\frac{\delta}{\kappa}$. Since
\begin{eqnarray*}
\lim_{\eta\rightarrow 0}f(\eta)=1,\qquad\lim_{\eta\rightarrow 0}f'(\eta)=-\infty,
\end{eqnarray*}
we have $f(\eta)<1$ for sufficiently small $\eta=\eta(Z)>0$. The contradiction implies the lemma.
\end{proof}

\noindent\textbf{Proof of Theorem \ref{m31} B.} By Lemma \ref{l33}, $\lambda(E^*)<\mu$. Let $\sigma_n$ be the number of $n$-step SAWs on $G$ starting from a fixed vertex such that one of the followings holds
\begin{itemize}
\item $E^*(m)$ occurs less than $an$ times; or
\item $E^*(m)$ occurs at least $an$ times, and $F$ occurs less than $a'n$ times
\end{itemize}
where $a$ satisfies (\ref{aml}) and $a'$ satisfies Lemma \ref{lf}. We have
\begin{eqnarray*}
\limsup_{n\rightarrow\infty}\sigma_n<\mu.
\end{eqnarray*}
For any $n$-step SAW $\pi$ on $G$ starting from a fixed vertex not counted in $\sigma_n$, $F$ occurs at least $a'n$ times, and therefore $\|\pi\|\geq a'n$. This implies that SAWs on $G$ have positive speed.
$\hfill\Box$

\section{Groups with more than one end}\label{p2}

In this section, we prove Theorem \ref{mg}. The proof is based on the stalling's splitting theorem, and an explicit construction of the set $S$ satisfying Assumption \ref{ap31}.

\begin{lemma}\label{l41}Let $G=(V,E)$ be an infinite, connected, locally finite graph. Let $A,B\subseteq V$ be two finite set of vertices of $G$ satisfying $A\subseteq B$. Let $G\setminus A$ (resp.\ $G\setminus B$) be the subgraph of $G$ obtained from $G$ by removing all the vertices in $A$ (resp.\ $B$) as well as their incident edges. If $G\setminus A$ has at least two infinite components, then $G\setminus B$ has at least two infinite components. Moreover, each infinite component of $G\setminus A$ contains at least one infinite component of $G\setminus B$.
\end{lemma}

\begin{proof}Since $G$ is a locally finite graph, we may assume that $G$ has maximal vertex degree $d$, where $1\leq d<\infty$ is a positive integer.
Let $R_1$ and $R_2$ be two infinite components of $G\setminus A$. We will show that each one of $R_1\setminus B$ and $R_2\setminus B$ has at least one infinite component. 

Since $R_1\setminus B$ can be obtained from $R_1$ by removing finitely many vertices and edges, if $R_1\setminus B$ has no infinite components, then $R_1\setminus B$ has infinitely many finite components. Moreover, since each vertex of $G$ is incident to at most $d$ edges, for any connected subgraph of $G$, removing one vertex of the subgraph as well as all its incident edges can split the subgraph into at most $d$ connected components. Since $|B|<\infty$, it is not possible that $R_1\setminus B$ has infinitely many finite components. As a result, $R_1\setminus B$ has at least one infinite component. The fact that $R_2\setminus B$ has at least one infinite component can be proved similarly.
\end{proof}

Now we prove Theorem \ref{mg}. By Stalling's splitting theorem (\cite{JS68,JS71}), a group $\Gamma$ has more than one end if and only if one of the followings holds. 
\begin{enumerate}
\item the group $\Gamma$ is an amalgated free product, i.e.
\begin{eqnarray*}
 \Gamma=H*_{C}K,
 \end{eqnarray*}
 where $H,K$ are groups, and $C$ is a finite group such that $C\neq H$ and $C\neq K$. More precisely, there exist group homomorphisms $\phi:C\rightarrow H$ and $\psi: C\rightarrow K$, such that $H*_{C}K$ can be obtained from the free product $H*K$ by adding relations $\phi(c)\psi^{-1}(c)$ for every $c\in C$. For simplicity, we shall identify $c$, $\phi(c)$ and $\psi(c)$ in all the remaining parts of the paper.
\item There exists a group $H$, two finite subgroups $F_1$, $F_2$ of $H$, and a group isomorphism $\phi: F_1\rightarrow F_2$, such that $\Gamma$ is the following HNN extension
\begin{eqnarray*}
\Gamma=\langle H,t|t^{-1}f_1t=\phi(f_1),\ \forall f_1\in F_1\rangle
\end{eqnarray*}
\end{enumerate}

We will prove Theorem \ref{mg} for Case (1) and Case (2), in the subsections below.

\subsection{Proof of Theorem \ref{mg} when $\Gamma$ is an amalgamated free product}\label{pf141}

In this section, we prove Theorem \ref{mg} when the group $\Gamma$ is a free product with amalgamation, as described in Case (1). We start with the following standard result concerning members in a amalgamated free product.

\begin{lemma}\label{l23}(Normal form for amalgamated free product \cite{LS77}) Every element in $H*_{C}K$ which is not in the image of $C$ can be written in the normal form
\begin{eqnarray*}
v_1\cdot\ldots\cdot v_n
\end{eqnarray*}
where the terms $v_i$ lie in $H\setminus C$ or $K\setminus C$ and alternate between two sets. The length $n$ is uniquely determined and two such expressions $v_1\cdot\ldots\cdot v_n$ give the same element in $H*_C K$ if and only if there are elements $c_1,\ldots, c_n\in C$, so that
\begin{eqnarray*}
w_k=c_{k-1}v_k c_k^{-1},
\end{eqnarray*}
where $c_0=c_n=1$.
\end{lemma}

Assume that $\Gamma$ is a free product with amalgamation as described in Case (1). Let $G_{H}$ (resp.\ $G_{K}$) be a locally finite Cayley graph for $H$ (resp.\ K) with respect to a finite set of generators $T_H$ (resp.\ $T_K$). Let $G_0$ be a locally finite Cayley graph of $\Gamma$ constructed from $G_H$ and $G_K$ as follows.
\begin{letlist}
\item Construct the free product graph $G_H*G_K$ of $G_H$ and $G_K$. In other words $G_H*G_K$ is the Cayley graph for the free product $H*K$ with respect to the generator set $T_H\cup T_K$.
\item Glue vertices $u\in G_H*G_K$ and $w\in G_H*G_K$ if there exists a vertex $v\in G_H*G_K$ such that $v^{-1}u=v^{-1}w\in C$.
\end{letlist}
Let $E_0$ denote the edge set of $G_0$.

Let $1_{\Gamma}$ be the identity element of the group $\Gamma$.
Let $G$ be a locally finite Cayley graph of $\Gamma$ with respect to the generator set $T$ such that $|T|<\infty$, $T=T^{-1}$ and $1_{\Gamma}\notin T$. Let 
\begin{eqnarray}
D_0=\max_{t\in T, v\in \Gamma}\mathrm{dist}_{G_0}(v,vt),\label{dd}
\end{eqnarray}
where $\mathrm{dist}_{G_0}(\cdot,\cdot)$ is the graph distance in $G_0$.

Let
\begin{eqnarray}
C_1=\{v\in \Gamma,\mathrm{dist}_{G_0}(v,C)\leq D_0+1\}.\label{dc}
\end{eqnarray}
Then $C\subseteq C_1$, and $|C_1|<\infty$. Let $C_2$ be a finite set of vertices containing $C_1$, such that $C_2$ is connected in $G$ in the sense defined as in Section 1.

\begin{lemma}\label{luv}Assume that $u\in \Gamma$ has normal form starting from a term in $H\setminus C$; and that $v\in \Gamma$ has normal form starting from a term in $K\setminus C$, then any path in $G_0$ joining $1_{\Gamma}$ and $u^{-1}v$ must visit a point in $u^{-1}C$.
\end{lemma}

\begin{proof}Since $u\in \Gamma$ has a normal form starting from a term in $H\setminus C$, and $v\in \Gamma$ has a normal form starting from a term in $K\setminus C$, then the concatenation of normal forms of $u^{-1}$ and $v$ gives us a normal form of $u^{-1}v$, denoted by $z_1\cdot\ldots \cdot z_n$. This normal form gives rise to a path $\tau_{u^{-1}v}$ in $G_0$ joining $1_{\Gamma}$ and $u^{-1}v$. More precisely, $\tau_{u^{-1}v}$ is the concatenation of $n$ paths $\tau_i$, $1\leq i\leq n$, such that $\tau_i$ is the shortest path in $z_1\cdot\ldots\cdot z_{i-1}H$ (resp.\ $z_1\cdot\ldots\cdot z_{i-1}K$) consisting of edges of $G_0$ and joining $z_1\cdot\ldots \cdot z_{i-1}$ and $z_1\cdot\ldots \cdot z_{i}$, if $z_i\in H\setminus C$ (resp. $z_i\in K\setminus C$). We can see that the path $\tau_{u^{-1}v}$ visits the vertex $u^{-1}$.

Assume that $\ell_{u^{-1}v}$ is an arbitrary path in $G_0$ joining $1_{\Gamma}$ and $u^{-1}v$. Then $\ell_{u^{-1}v}$ gives rise to a sequence $x_1,\ldots,x_m$, such that
\begin{letlist}
\item for $1\leq i\leq m$, $\ell_{u^{-1}v}$ visits every vertex in $\{x_1\cdot\ldots\cdot x_i\}_{1\leq i\leq m}$;
\item for $1\leq i\leq m$, $x_i\in H\setminus C$ or $x_i\in K\setminus C$;
\item for $1\leq i\leq m-1$, if $x_i\in H\setminus C$, then $x_{i+1}\in K\setminus C$;
\item $u^{-1}v=x_1\cdot\ldots \cdot x_m c$, where $c\in C$.
\end{letlist}
Indeed $x_1,\ldots, x_m$ can be found by the following induction process. Let $\{y_1,\ldots, y_{\ell}\}$ be all the vertices incident to an edge in $[H\setminus C]\cap E_0$ along $\ell_{u^{-1}v}$ or an edge in $[K\setminus C]\cap E_0$ along $\ell_{u^{-1}v}$, and assume that starting from $1_{\Gamma}$ and traversing along $\ell_{u^{-1}v}$, one visits $y_1,\ldots,y_{\ell}$ in order. Let $y_0=1_{\Gamma}$. From $\{y_1,\ldots,y_{\ell}\}$, we perform the following manipulations.
\begin{Alist}
\item Remove all the $y_i$'s such that $y_{i-1}^{-1}y_i\in C$; let $\{y_1^{(1)},\ldots,y_{\ell_{1}}^{(1)}\}$ be the remaining set of vertices;
\item Remove all the $y_i^{(1)}$'s such that either both $[y_{i-1}^{(1)}]^{-1}y_{i}^{(1)}$ and $[y_{i}^{(1)}]^{-1}y_{i+1}^{(1)}$ are in $H$ or both $[y_{i-1}^{(1)}]^{-1}y_{i}^{(1)}$ and $[y_{i}^{(1)}]^{-1}y_{i+1}^{(1)}$ are in $K$; let $\{y_1^{(2)},\ldots,y_{\ell_{2}}^{(2)}\}$ be the remaining set of vertices.
\end{Alist}
Once we have constructed $\{y_1^{(2j)},\ldots,y_{\ell_{2j}}^{(2j)}\}$, for $j\geq 1$, we perform the following manipulations.
\begin{Alist}
\item Remove all the $y_i$'s  such that $[y_{i-1}^{(2j)}]^{-1}y_i^{(2j)}\in C$; let $\{y_1^{(2j+1)},\ldots,y_{\ell_{1}}^{(2j+1)}\}$ be the remaining set of vertices;
\item Remove all the $y_i^{(1)}$'s such that either both $[y_{i-1}^{(2j+1)}]^{-1}y_{i}^{(2j+1)}$ and $[y_{i}^{(2j+1)}]^{-1}y_{i+1}^{(2j+1)}$ are in $H$ or both $[y_{i-1}^{(2j+1)}]^{-1}y_{i}^{(2j+1)}$ and $[y_{i}^{(2j+1)}]^{-1}y_{i+1}^{(2j+1)}$ are in $K$; let $\{y_1^{(2j+2)},\ldots,y_{\ell_{2j+2}}^{(2j+2)}\}$ be the remaining set of vertices.
\end{Alist}
We repeat the process above until we end up with a set of vertices $\{y_1^{(2k)},\ldots,y_{\ell_{2k}}^{(2k)}\}$ satisfying 
\begin{letlist}
\item for $1\leq i\leq \ell_{2k}$, $[y_{i-1}^{(2k)}]^{-1}y_i^{(2k)}\in H\setminus C$ or $[y_{i-1}^{(2k)}]^{-1}y_i^{(2k)}\in K\setminus C$;
\item for $1\leq i\leq \ell_{2k}-1$, if $[y_{i-1}^{(2k)}]^{-1}y_i^{(2k)}\in H\setminus C$, then $[y_{i}^{(2k)}]^{-1}y_{i+1}^{(2k)}\in K\setminus C$;
\item for $1\leq i\leq \ell_{2k}-1$, if $[y_{i-1}^{(2k)}]^{-1}y_i^{(2k)}\in K\setminus C$, then $[y_{i}^{(2k)}]^{-1}y_{i+1}^{(2k)}\in H\setminus C$;
\item $u^{-1}v=y_{\ell_{2k}}^{(2k)} c$, where $c\in C$.
\end{letlist}
Obviously the process above will terminate in finitely many steps. Let $x_1=y_1^{(2k)}$, for $2\leq i\leq \ell_{2k}$, $x_{i}=[y_{i-1}^{(2k)}]^{-1}y_i^{(2k)} $. Then 
\begin{eqnarray*}
x_1\cdot x_2\cdot\ldots \cdot [x_{\ell_{2k}}c]
\end{eqnarray*}
gives another normal form for $u^{-1}v$. By the uniqueness of normal form Lemma \ref{l23}, we have $\ell_{2k}=m=n$, and 
\begin{eqnarray*}
z_k=c_{k-1}^{-1}x_k c_{k}.
\end{eqnarray*}
where $c_k\in C$.
Since $\ell_{u^{-1}v}$ visits every vertex in $\{x_1\cdot\ldots\cdot x_{i}\}$, $1\leq i\leq n$, $\ell_{u^{-1}v}$ visits a vertex in $u^{-1}C$.
\end{proof}

\begin{lemma}\label{l43}Let $u\in \Gamma$ have normal form starting from a term in $H\setminus C$; and let $v\in \Gamma$ have normal form starting from a term in $K\setminus C$. If the lengths of normal forms for $u$ and $v$ exceed the maximal length of normal forms of vertices in $C_1$ (resp.\ $C_2$), $u$ and $v$ are in two distinct infinite components of $G\setminus C_1$ (resp.\ $G\setminus C_2$).

In particular, $G\setminus C_1$ (resp.\ $G\setminus C_2$) has at least two distinct infinite components.
\end{lemma}
\begin{proof}By Lemma \ref{luv}, any path in $G_0$ joining $1_{\Gamma}$ and $u^{-1}v$ must visit a point in $u^{-1}C$.

Assume that there exists a path $\ell_{uv}$ joining $u$ and $v$ in $G\setminus C_1$. Then for any edge $\langle p,q \rangle$ along $\ell_{uv}$, $p$ and $q$ can be joined by a path in $G_0$ which does not pass through any vertex in $C$. That is because $\mathrm{dist}_{G_0}(p,q)\leq D_0$ by (\ref{dd}), and $\mathrm{dist}_{G_0}(p,C)\geq D_0+2$ and $\mathrm{dist}_{G_0}(q,C)\geq D_0+2$, by (\ref{dc}) and the fact that $p,q\in G\setminus C_1$. Consider the shortest path in $G_0$ joining $p$ and $q$, the length of this path is at most $D_0$. If this path passes through a vertex in $C$, then the length of this path is at least $2D_0+4$. The contradiction implies that the shortest path in $G_0$ joining $p$ and $q$ does not pass through any vertex in $C$.

By replacing each edge $\langle p,q \rangle$ along $\ell_{uv}$ with a path in $G_0$ joining $p$ and $q$ that does not pass through any vertex in $C$, we obtain a path in $G_0$ joining $u$ and $v$ without passing through any vertex in $C$. This implies that there exists a path in $G_0$ joining $1_{\Gamma}$ and $u^{-1}v$ which does not visit any vertex of $u^{-1}C$, by the vertex-transitivity of $G_0$. This is a contradiction.

As a result, any path in $G$ joining $u$ and $v$ must visit a vertex in $C_1$. Since $C_1$ is finite, the lengths of normal forms for vertices in $C_1$ have a maximum $M_0$. If the lengths of normal forms for $u$ and $v$ exceed $M_0$, then $u\in G\setminus C_1$ and $v\in G\setminus C_1$. This means that $u$ and $v$ are in two distinct components of $G\setminus C_1$. Moreover, we can find a singly-infinite path $\ell_u$ (resp.\ $\ell_v$) on $G$ starting from $u$ (resp.\ $v$), such that moving along the path from $u$, the length of normal forms along the path is non-decreasing. Therefore all the vertices along $\ell_u$ (resp.\ $\ell_v$) are in $G\setminus C_1$. As a result, $u$ and $v$ are in two distinct infinite component of $G\setminus C_1$. Similar arguments applies if we replace $C_1$ by $C_2$.
\end{proof}

Let $\Omega_n$ be the set of $n$-step SAWs on $G$ starting from the identity vertex $1_{\Gamma}$. We make the following assumptions
\begin{assumption}\label{ap45}
\begin{enumerate}
\item Let $r_1\in \Gamma$ have a normal form starting from a term in $H\setminus C$, and ending in a term in $H\setminus C$.  
\item Let $r_2\in \Gamma$ have a normal form starting from a term in $K\setminus C$ and ending at a term in $K\setminus C$. 
\item The lengths of normal forms of $r_1$ and $r_2$ are strictly greater than the maximal length of normal forms of elements in $C_2$. 
\item 
\begin{eqnarray}
\mathrm{dist}_{G}(r_1,C_2)\geq |C_2|+1;\label{dr1}\\
\mathrm{dist}_{G}(r_2,C_2)\geq |C_2|+1.\label{dr2}
\end{eqnarray}
\end{enumerate}
\end{assumption}

Let $A$ (resp.\ $B$) be the collection of elements in $\Gamma$ with a normal form starting from an element in $H\setminus C$ (resp.\ $K\setminus C$).
Let
\begin{eqnarray}
S=C_2,\label{cs}
\end{eqnarray}
then by the construction of $C_2$ it is obvious that $S$ is connected. The graph $G\setminus S$ has at least two distinct infinite components by Lemma \ref{l43}, $C_1\subset C_2$ and Lemma \ref{l41}. Let $A_1\subseteq A$ be a component of $G\setminus S$, define $\phi(S,A_1)=r_2$; i.e. $\phi(S,A_1)$ is the graph automorphism such that for any $x\in \Gamma$, $\phi(S,A_1)x=r_2 x$. Let $Q_1\subseteq B$ be a component of $G\setminus S$, define $\phi(S,Q_1)=r_1$.

\begin{lemma}\label{l46}$r_2A_1\cap S=\emptyset$.
\end{lemma}
\begin{proof}Recall that $A_1$ is a component of $G\setminus S$ such that all the elements in $A_1$ have a normal form starting from a term in $H\setminus C$.  For each $v\in A_1$, since $r_2$ has a normal form ending in a term in $K\setminus C$, and $v$ has a normal from starting from a term in $H\setminus C$, the concatenation of normal forms of $r_2$ and $v$ gives us a normal form of $r_2v$. Under Assumption \ref{ap45} (3), the length of the normal form of $r_2v$ is also strictly greater than the maximal length of normal forms of elements in $S$, we have $r_2 v\notin S$. Since $v$ is an arbitrary point in $A_1$. By Lemma $\ref{l43}$ we obtain $r_2 A_1\cap S=\emptyset$.
\end{proof}

\begin{lemma}\label{l47}$r_2A_1$ is in a connected component of $G\setminus S$ different from $A_1$.
\end{lemma}

\begin{proof}For any two vertices $u,v\in A_1$, by the connectivity of $A_1$, there exists a path $\ell_{uv}$ joining $u$ and $v$ and consisting of vertices in $A_1$. Then $r_2\ell_{uv}$ is a path joining $r_2u$ and $r_2v$ which does not intersect $S$, since any vertex along $r_2 \ell_{uv}$ is in $r_2 A_1$, and $r_2 A_1\cap S=\emptyset$ by Lemma \ref{l46}. Therefore all the vertices in $r_2 A_1$ are in the same component $B_1$ of $G\setminus S$.

Moreover, for each $w\in A_1$, $r_2w$ has a normal form starting from a term in $K\setminus C$ by Assumption \ref{ap45} (2), and all the elements in $A_1$ have a normal form starting from a term in $H\setminus C$, we have $r_2 A_1\cap A_1=\emptyset$. By Assumption \ref{ap45}(3), the length of normal form of each element in $r_2A_1$ exceeds the maximal length of normal forms of elements in $C_2$.  By Lemma \ref{l43}, $r_2 A_1$ and $A_1$ are in two distinct components of $G\setminus C_2=G\setminus S$.
\end{proof}

\begin{lemma}Assumption \ref{ap31} (3) holds with the choice of $S$ as given by (\ref{cs}).
\end{lemma}
\begin{proof}For each component of $G\setminus S$, either all the elements in the component have
a normal form starting from a term in $H\setminus C$, or all the elements in the component
have a normal form starting from a term in $K \setminus C$, by Lemma \ref{l43}.
Recall that $A_1$ is an arbitrary component of $G \setminus S$ such that all the elements in $A_1$ have a normal
form starting from a term in $H \setminus C$; and $B_1$ is the component of $G\setminus S$ containing $r_2A_1$.

As in Assumption \ref{ap31} (3), let $\partial_{A_1}S$ be the set consisting of all the vertices in $S$ incident to a vertex in $A_1$. Let $u\in A_1$, $w\in \partial_{A_1}S$, and $e$ be the edge of $G$ with endpoints $u$ and $w$. By Lemmas \ref{l46} and \ref{l47}, it suffices to show that $r_2 w\in \partial_{B_1}S\cup B_1$, $w$ and $r_2 w$ are joined by a path in $G\setminus [A_1\cup r_2 A_1]$, whose length is bounded above by a constant $N$ independent of $A_1$ and $w$.

 Since $r_2 u\in B_1$, we have $r_2 w\in B_1\cup \partial_{B_1}S$, since $r_2 u$ and $r_2 w$ are adjacent vertices in $G$.  Let $l_{w,r_2 w}$ be a path in $G$ joining $w$ and $r_2 w$, starting from $w$ and ending in $r_2 w$. Let $p$ be the last vertex of $S$ visited  by $l_{w, r_2 w}$, and let $q$ be the first vertex of $r_2 S$ visited by $l_{w,r_2 w}$. Then $l_{w,r_2w}$ is divided by $p$ and $q$ into 3 portions: $l_{wp}$, $l_{pq}$ and $l_{q,r_2w}$.

 By the connectivity of $S$, there exists a path $L_{wp}$ joining $w$ and $p$ and consisting of vertices of $S$; also, there exists a path $L_{q, r_2w}$ joining $q$ and $r_2 w$ and consisting of vertices of $r_2S$.  We shall prove the following lemmas concerning $L_{wp}$, $L_{q,r_2w}$ and $l_{pq}$.
 
 \begin{lemma}\label{l49}$L_{wp}\in G\setminus[A_1\cup r_2 A_1]$; $L_{q,r_2w}\in G\setminus[A_1\cup r_2 A_1]$.
 \end{lemma}
 \begin{proof}Since $L_{wp}\in S$,  $r_2 A_1\subset B_1$, and $A_1$ and $B_1$ are two distinct connected components of $G\setminus S$, we have $L_{wp}\in G\setminus [A_1\cup r_2 A_1]$.

 By (\ref{dr2}), $r_2 S\cap S=\emptyset$. Since $r_2\in r_2 S$ and $r_2\notin A_1$, the connectivity of $r_2 S$ implies that $r_2 S$ is in a component of $G\setminus S$ different from $A_1$; in particular $r_2 S\cap A_1=\emptyset$. Moreover, since $r_2w\in r_2 S$, $r_2w\in B_1\cup \partial_{B_1}S$, we have $r_2 S\subset B_1$. 
 Hence $L_{q,r_2w}\subset r_2 S\subset B_1\subset G\setminus A_1$. Since $S\cap A_1=\emptyset$, $r_2 S\cap r_2 A_1=\emptyset$, we obtain that $L_{q,r_2w}\subset r_2 S\subset G\setminus [A_1\cup r_2 A_1]$.
 \end{proof}
 
 \begin{lemma}\label{l410}$[l_{pq}\setminus\{p,q\}]\in [G\setminus[A_1\cup r_2 A_1]]$.
 
 \begin{proof}Recall that $l_{pq}$ be the portion of $l_{w,r_2w}$ between $p$ and $q$. All the vertices along $l_{pq}$ except $p$ are outside $S$, hence they are in the same component of $G\setminus S$. Since $q\in r_2S\subset B_1$, all the vertices along $l_{pq}$ except $p$ are in $B_1$. Since $B_1\cap A_1=\emptyset$, we have
\begin{eqnarray} 
 [l_{pq}\setminus\{p\}]\subset [G\setminus (A_1)]\label{st1}
\end{eqnarray} 
  Similarly, all the vertices along $ l_{pq}$ except $q$ are outside  $r_2S$, hence they are in the same component of $G\setminus r_2 S$. Under the assumption that the length of the normal form of $r_2$ is strictly greater than the maximal length of normal forms of elements in $C_2 (:=S)$, we have $r_2 A_1\cap S=\emptyset$.  Since $p\in S$, $p\notin r_2 A_1$. Note that $r_2A_1$ is a connected component of $G\setminus r_2 S$. Hence
 \begin{eqnarray}
 [l_{pq}\setminus\{q\}]\subset [G\setminus (r_2A_1)]\label{st2}
 \end{eqnarray}

By (\ref{st1}) and (\ref{st2}), no vertices in $\ell_{pq}$ except $p$ and $q$ are in $A_1\cup r_2 A_1$.
\end{proof}
\end{lemma} 
 
 Let $L_{w,r_2w}=L_{wp}\cup l_{pq}\cup L_{q,r_2w}$; then $L_{w,r_2w}$ is a path joining $w$ and $r_2w$ in $G\setminus [A_1\cup r_2 A_1]$ by Lemmas \ref{l49} and \ref{l410}.
\end{proof}

\bigskip

\subsection{Proof of Theorem \ref{t15} when $\Gamma$ is an amalgamated free product}

In this section, we prove Theorem \ref{t15} when $\Gamma$ is an amalgamated free product. Let $\Gamma$ be a finitely generated, infinite group which is a free product with amalgamation as described in (1). It suffices to construct a locally finite Cayley graph $G$ of $\Gamma$ such that SAWs on $G$ have positive speed.

Choose a Cayley graph $G_H$ (resp.\ $G_K$) for $H$ (resp.\ $K$) such that any two vertices in $C$ (resp.\ $C$) are joined by an edge. Let $G$ be the graph obtained from the free product graph $G_H*G_K$ by gluing the vertices $u\in G_{H}*G_{K}$ and $w\in G_{H}*G_{K}$ satisfying the condition that there exists a vertex $v\in G_{H}*G_{K}$ such that $v^{-1}u=v^{-1}w\in C$.

Let 
\begin{eqnarray*}
S'=S=C
\end{eqnarray*}
It is not hard to check that for the locally finite Cayley graph $G$ of $\Gamma$ constructed above with the finite set of vertices $S$, Assumption \ref{ap32} is satisfied and SAWs on $G$ have positive speed.

\subsection{HNN extension}

In this section, we prove Theorems \ref{mg} and \ref{t15} when $\Gamma$ is an HNN extension as described by Part (2) of the Stalling's splitting theorem. Again we shall explicitly construct the ``cut sets'' $S$ and $S'$ satisfying Assumptions \ref{ap31} and \ref{ap32} based on the structures of the group.

Let $F_1$ and $F_2$ be the two finite subgroups of $H$ as in Part (2) of the Stalling's splitting theorem.
Choose a set $R_1$ of representatives of the right cosets of $F_1$ in $H$, and a set $R_2$ of representatives of the right cosets of $F_2$ in $H$.  We shall assume that the identity element $1$ of $H$ is in both $R_1$ and $R_2$. In particular, $R_1$ (resp.\ $R_2$) is a subset of $H$ whose elements are in 1-1 correspondence with right cosets of $F_1$ (resp.\ $F_2$) in $H$.
 The choice of coset representatives is to be fixed to the rest of the discussion.

\begin{definition}(Normal form for HNN extension \cite{LS77})\label{dfn} Let $\Gamma$ be the HNN extension with a presentation
\begin{eqnarray*}
\Gamma=\langle H,t|t^{-1}f_1t=\phi(f_1),\ \forall f_1\in F_1 \rangle.
\end{eqnarray*}
where $H$ is a group, $F_1$, $F_2$ are two finite subgroups of $H$, and $\phi: F_1\rightarrow F_2$ is a group isomorphism.
 A normal form is a sequence $g_0, t^{\epsilon_1},g_1,t^{\epsilon_2},\ldots,t^{\epsilon_{n-1}},g_{n-1},t^{\epsilon_n},g_n (n\geq 0)$ where
\begin{itemize}
\item $g_0$ is an arbitrary element of $\Gamma$;
\item for $1\leq i\leq n$, if $\epsilon_i=-1$, then $g_i\in R_1$;
\item for $1\leq i\leq n$, if $\epsilon_i=1$, then $g_i\in R_2$; and
\item there is no consecutive subsequence $t^{\epsilon},1,t^{-\epsilon}$.
\end{itemize}
\end{definition}

\begin{theorem}(Uniqueness of normal form \cite{LS77})\label{t47}Every element $w$ of $\Gamma$ with a presentation as in (2) has a unique representation as
\begin{eqnarray*}
w=g_0t^{\epsilon_1}\cdot\ldots\cdot t^{\epsilon_n}g_n,
\end{eqnarray*}
where $g_0,t^{\epsilon_1},\ldots,t^{\epsilon_n},g_n$ is a normal form. Let $n$ be the \textbf{length} of the normal form of $w$. 
\end{theorem}

\subsubsection{Proof of Theorem \ref{mg} when $\Gamma$ is an HNN extension}
Let $\Gamma$ be an infinite, finitely generated group, which is an HNN extension as described in Part (2) of the Stalling's splitting theorem. Let $G_H$ be a locally finite Cayley graph for $H$ with respect to a finite set of generators $T_H$ satisfying $|T_H|<\infty$, $T_H=T_H^{-1}$ and $1\notin T_H$.

Let $G_0$ be the Cayley graph of the HNN extension $\Gamma$ with respect to the generator set $T_H\cup\{t,t^{-1}\}$. Let $G$ be a locally finite Cayley graph of $\Gamma$ with respect to a finite generator set $T$ satisfying $T=T^{-1}$, $|T|<\infty$ and $1\notin T$.

Let
\begin{eqnarray}
D_0&=&\max_{v\in \Gamma,s\in T}\mathrm{dist}_{G_0}(v,vs);\label{dd0}\\
S_0&=&\left\{w\in\Gamma,\mathrm{dist}_{G_0}(w,F_1\cup F_2)\leq \frac{D_0}{2}\right\}\label{s0}
\end{eqnarray}
Obviously $S_0$ is finite since both $F_1$ and $F_2$ are finite.

\begin{lemma}\label{css}Let $u,v\in \Gamma$ such that one  of the followings hold
\begin{letlist}
\item $v$ has a normal form with $\epsilon_1=1$, $g_0\in F_1$; and $u$ has a normal form with $\epsilon_1=-1$; or 
\item $v$ has a normal form with $\epsilon_1=-1$, $g_0\in F_2$; and $u$ has a normal form with $\epsilon_1=1$; or
\item $v$ has a normal form with $\epsilon_1=1$, $g_0\in F_1$; and $u$ has a normal form with $\epsilon_1=1$, $g_0\in H\setminus F_1$; or
\item $v$ has a normal form with $\epsilon_1=-1$, $g_0\in F_2$; and $u$ has a normal form with $\epsilon_1=-1$, $g_0\in H\setminus F_2$;
\end{letlist}
Then any path in $G_0$ joining $1_{\Gamma}$ and $u^{-1}v$ must visit a vertex in $u^{-1}[F_1\cup F_2]$.
\end{lemma}
\begin{proof}It suffices to prove that any self-avoiding path in $G_0$ joining $1_{\Gamma}$ and $u^{-1}v$ must visit a vertex in $u^{-1}[F_1\cup F_2]$.

Let $\ell_{u^{-1}v}$ be an arbitrary self-avoiding path in $G_0$ joining $1_{\Gamma}$ and $u^{-1}v$.  Let $y_1,\ldots,y_m$ be all the vertices along $\ell_{u^{-1}v}$ such that one of the followings holds
\begin{Alist}
\item both edges incident to $y_i$ along $\ell_{u^{-1}v}$ are $t$ or $t^{-1}$; or
\item one edge incident to $y_i$ along $\ell_{u^{-1}v}$ is $t$ or $t^{-1}$; the other edge incident to $y_i$ along $\ell_{u^{-1}v}$ is an edge of $G_{H}$.
\end{Alist}
Count each vertex in Case A. twice in $\{y_1,\ldots, y_m\}$, and count each vertex in Case B. once in $\{y_1,\ldots,y_m\}$. More precisely, if $x$ is a vertex along $\ell_{u^{-1}v}$ such that both edges incident to $x$ along $\ell_{u^{-1}v}$ are $t$ or $t^{-1}$, then there exists $1\leq i\leq m-1$, such that $y_i=y_{i+1}=x$. Moreover, if $z$ is a vertex along $\ell_{u^{-1}v}$ such that  one edge incident to $z$ along $\ell_{u^{-1}v}$ is $t$ or $t^{-1}$; the other edge incident to $z$ along $\ell_{u^{-1}v}$ is an edge of $G_{H}$, then there exists exactly one $i\in\{1,2,\ldots,m\}$, such that $y_i=z$.

We perform the following manipulations on the set $\{y_1,\ldots,y_m\}$.
\begin{enumerate}
\item For all $i$'s satisfying $y_{i-1}^{-1}y_i=t^{-1}$, $y_i^{-1}y_{i+1}\in F_1$, and $y_{i+1}^{-1}y_{i+2}=t$, remove $y_{i-1},y_i,y_{i+1},y_{i+2}$ from $\{y_1,\ldots,y_m\}$, and let $\{y_1^{(1)},\ldots,y_{n_1}^{(1)}\}$ be the new sequence;
\item For all $i$'s satisfying $[y_{i-1}^{(1)}]^{-1}y_i^{(1)}=t$, $[y_i^{(1)}]^{-1}y_{i+1}^{(1)}\in F_2$, and $[y_{i+1}^{(1)}]^{-1}y_{i+2}^{(1)}=t^{-1}$, remove $y_{i-1}^{(1)},y_i^{(1)},y_{i+1}^{(1)},y_{i+2}^{(1)}$ from $\{y_1^{(1)},\ldots,y_{n_1}^{(1)}\}$, and let $\{y_1^{(2)},\ldots,y_{n_2}^{(2)}\}$ be the new sequence;
\end{enumerate}
Assume we have obtained $\{y_1^{(2j)},\ldots,y_{n_{2j}}^{(2j)}\}$, then we perform the following inductive manipulations
\begin{enumerate}
\item For all $i$'s satisfying $[y_{i-1}^{(2j)}]^{-1}y_i^{(2j)}=t^{-1}$, $[y_i^{(2j)}]^{-1}y_{i+1}^{(2j)}\in F_1$, and $[y_{i+1}^{(2j)}]^{-1}y_{i+2}^{(2j)}=t$, remove $y_{i-1}^{(2j)},y_i^{(2j)},y_{i+1}^{(2j)},y_{i+2}^{(2j)}$ from $\{y_1^{(2j)},\ldots,y_{n_{2j}}^{(2j)}\}$, and let $\{y_1^{(2j+1)},\ldots,y_{n_{2j+1}}^{(2j+1)}\}$ be the new sequence;
\item For all $i$'s satisfying $[y_{i-1}^{(2j+1)}]^{-1}y_i^{(2j+1)}=t$, $[y_i^{(2j+1)}]^{-1}y_{i+1}^{(2j+1)}\in F_2$, and $[y_{i+1}^{(2j+1)}]^{-1}y_{i+2}^{(2j+1)}=t^{-1}$, remove $y_{i-1}^{(2j+1)},y_i^{(2j+1)},y_{i+1}^{(2j+1)},y_{i+2}^{(2j+1)}$ from $\{y_1^{(2j+1)},\ldots,y_{n_{2j+1}}^{(2j+1)}\}$, and let $\{y_1^{(2j+2)},\ldots,y_{n_{2j+2}}^{(2j+2)}\}$ be the new sequence.
\end{enumerate}
We continue the above process until we obtain $\{y_1^{(2k)},\ldots,y_{n_{2k}}^{(2k)}\}$ such that 
\begin{romlist}
\item There are no $i$'s satisfying $[y_{i-1}^{(2k)}]^{-1}y_i^{(2k)}=t^{-1}$, $[y_i^{(2k)}]^{-1}y_{i+1}^{(2k)}\in F_1$, and $[y_{i+1}^{(2k)}]^{-1}y_{i+2}^{(2k)}=t$; and 
\item There are no $i$'s satisfying $[y_{i-1}^{(2k)}]^{-1}y_i^{(2k)}=t$, $[y_i^{(2k)}]^{-1}y_{i+1}^{(2k)}\in F_2$, and $[y_{i+1}^{(2k)}]^{-1}y_{i+2}^{(2k)}=t^{-1}$.
\end{romlist}
Obviously the above inductive process will terminate after finitely many steps. Then we have 
\begin{eqnarray*}
&&y_1^{(2k)}=\eta_0;\ y_2^{(2k)}=\eta_0t^{\xi_1};\ y_3^{(2k)}=\eta_0t^{\xi_1}\eta_1;\ \ldots;\\
&&y^{(2k)}_{n_{2k}}=\eta_0t^{\xi_1}\ldots t^{\xi_n};\ u^{-1}v= y^{(2k)}_{n_{2k}}\eta_n.
\end{eqnarray*}
where $\eta_i\in H$. Since $\{y_1^{(2k)},\ldots,y_{n_{2k}}^{(2k)}\}\subset\{y_1,\ldots,y_{m}\}$, all the vertices in $\{y_1^{(2k)},\ldots,y_{n_{2k}}^{(2k)}\}$ are visited by $\ell_{u^{-1}v}$. Working from the right to the left of $\eta_0t^{\xi_1}\ldots t^{\xi_n}\eta_n$, we can change it to a normal form $\theta_0t^{\xi_1}\ldots t^{\xi_n}\theta_n$, such that for $1\leq i\leq n$, $\theta_i\in R_1$ (resp.\ $\theta_i\in R_2$) if $\xi_i=-1$ (resp.\ $\xi_i=1$), as explained in Definition \ref{dfn}. More precisely, we find $\theta_1,\ldots,\theta_n$ in the following way
\begin{enumerate}
\item \begin{itemize}
\item If $\xi_n=-1$, choose $\theta_n\in R_1$, such that $F_1 \theta_n=F_1\eta_n$. Then 
\begin{eqnarray*}
t^{-1}\eta_n=[t^{-1}\eta_n\theta_n^{-1}t]t^{-1}\theta_n;
\end{eqnarray*}
where $t^{-1}\eta_n\theta_n^{-1}t\in F_2$. Let $q_n:=t^{-1}\eta_n\theta_n^{-1}t$.
\item If $\xi_n=1$, choose $\theta_n\in R_2$, such that $F_2 \theta_n=F_2\eta_n$. Then 
\begin{eqnarray*}
t\eta_n=[t\eta_n\theta_n^{-1}t^{-1}]t\theta_n;
\end{eqnarray*}
where $t\eta_n\theta_n^{-1}t^{-1}\in F_1$. Let $q_n:=t\eta_n\theta_n^{-1}t^{-1}$.
\end{itemize}
\item Let $1\leq i\leq n-1$. Assume we have determined $\theta_n,\ldots, \theta_{i+1}$ and $q_n,\ldots,q_{i+1}$. Then
\begin{itemize}
\item If $\xi_i=-1$, choose $\theta_i\in R_1$ such that $F_1\theta_i=F_1\eta_i q_{i+1}$, then
\begin{eqnarray*}
t^{-1}\eta_i q_{i+1}=[t^{-1}\eta_i q_{i+1}\theta_i^{-1}t]t^{-1}\theta_i
\end{eqnarray*}
where $t^{-1}\eta_i q_{i+1}\theta_i^{-1}t=F_2$. Let $q_i:=t^{-1}\eta_i q_{i+1}\theta_i^{-1}t$.
\item If $\xi_i=1$, choose $\theta_i\in R_1$ such that $F_2\theta_i=F_2\eta_i q_{i+1}$, then
\begin{eqnarray*}
t\eta_i q_{i+1}=[t \eta_i q_{i+1}\theta_i^{-1}t^{-1}]t\theta_i
\end{eqnarray*}
where $t\eta_i q_{i+1}\theta_i^{-1}t^{-1}=F_1$. Let $q_i:=t\eta_i q_{i+1}\theta_i^{-1}t^{-1}$.
\end{itemize}
\item Let $\theta_0=\eta_0 q_1$.
\end{enumerate}

The proof of Lemma \ref{css} makes use of the following two lemmas.

\begin{lemma}
Assume that $u$ and $v$ satisfy one of (a),(b),(c),(d) as in the statement of the Lemma \ref{css}. The normal form $\theta_0t^{\xi_1}\ldots t^{\xi_n}\theta_n$ of $u^{-1}v$ gives rise to a path by using a path $\tau_i$ ($0\leq i\leq n$) in $\theta_0t^{\xi_1}\ldots t^{\xi_{i}}H$ to join $\theta_0t^{\xi_1}\ldots t^{\xi_{i}}$ and $\theta_0t^{\xi_1}\ldots t^{\xi_{i}}\theta_i$, and concatenating these path as well as the $t$-edge or $t^{-1}$-edge joining them. Then the path must visit a vertex in $u^{-1}[F_1\cup F_2]$. 

\end{lemma}

\begin{proof}It is straightforward to check that if $u$ and $v$ satisfy one of (a), (b), (c), (d) as in the statement of Lemma \ref{css}, then the concatenation of normal forms of $u^{-1}$ and $v$ gives rise to a normal form of $u^{-1}v$. By the uniqueness of the normal form as stated in Lemma \ref{t47}, the concatenation of normal forms of $u^{-1}$ and $v$ is exactly the normal form $\theta_0t^{\xi_1}\ldots t^{\xi_n}\theta_n$ of $u^{-1}v$, and they give rise to the same path in $G_0$. Note that in Cases (a)(c) of Lemma \ref{css}, the path visits a vertex in $u^{-1}F_1$; while in Cases (b) (d) of Lemma \ref{css}, the path visits a vertex in $u^{-1}F_2$.
\end{proof}

\begin{lemma}
One vertex in $\{y_1^{(2k)},\ldots,y_{n_{2k}}^{(2k)}\}$ is in $u^{-1}[F_1\cup F_2]$. 
\end{lemma}

\begin{proof}Reviewing the process of constructing the normal form $\theta_0t^{\xi_1}\ldots t^{\xi_n}\theta_n$ from $\{y_1^{(2k)},\ldots,y_{n_{2k}}^{(2k)}\}$, one can find that for $1\leq i\leq n$,
\begin{eqnarray*}
\theta_0t^{\xi_1}\ldots t^{\xi_{i}}\theta_i=\eta_0t^{\xi_1}\ldots t^{\xi_{i}}\eta_i q_{i+1}.
\end{eqnarray*}
where $q_{n+1}:=1$. We consider the following cases
\begin{itemize}
\item In Cases (a)(c) of Lemma \ref{css}, since $v$ has a normal form with $\epsilon_1=1$ and $g_0\in F_1$, there exists $1\leq j\leq i-1$, such that
\begin{eqnarray*}
\theta_0t^{\xi_1}\ldots t^{\xi_{j}}\theta_j\in u^{-1}F_1
\end{eqnarray*}
with $\xi_{j+1}=1$. Hence $q_{j+1}\in F_1$, and therefore
\begin{eqnarray*}
\eta_0t^{\xi_1}\ldots t^{\xi_{j}}\eta_j=\theta_0t^{\xi_1}\ldots t^{\xi_{j}}\theta_j q_{j+1}^{-1}\in u^{-1}F_1.
\end{eqnarray*}
\item In Cases (b)(d) of Lemma \ref{css}, since $v$ has a normal form with $\epsilon_1=-1$ and $g_0\in F_2$, there exists $1\leq j\leq i-1$, such that
\begin{eqnarray*}
\theta_0t^{\xi_1}\ldots t^{\xi_{j}}\theta_j\in u^{-1}F_2
\end{eqnarray*}
with $\xi_{j+1}=-1$. Hence $q_{j+1}\in F_2$, and therefore
\begin{eqnarray*}
\eta_0t^{\xi_1}\ldots t^{\xi_{j}}\eta_j=\theta_0t^{\xi_1}\ldots t^{\xi_{j}}\theta_j q_{j+1}^{-1}\in u^{-1}F_2.
\end{eqnarray*}
\end{itemize}
Then the lemma follows.
\end{proof}

Since $\ell_{u^{-1}v}$ visits every vertex in $\{y_1^{(2k)},\ldots,y_{n_{2k}}^{(2k)}\}$, it must visit a vertex in $u^{-1}[F_1\cup F_2]$. Then the proof is complete.
\end{proof}

\begin{lemma}\label{l48}Let $S_0$ be defined as in (\ref{s0}). Let $G\setminus S_0$ be the subgraph obtained from $G$ by removing all the vertices in $S_0$ as well as their incident edges. For any $u,v\in G\setminus S_0$ one of (a) (b) (c) (d) in Lemma \ref{css} holds, $u$ and $v$ are in two distinct components of $G\setminus S_0$. Moreover, $G\setminus S_0$ has at least two infinite components.
\end{lemma}
\begin{proof}To show that $u$ and $v$ are in two distinct components of $G\setminus S_0$, it suffices to show that any path in $G$ joining $u$ and $v$ must visit a vertex in $S_0$.

 Let $\ell_{uv}$ be a path in $G$ joining $u$ and $v$. Assume that $\ell_{uv}$ visits no vertices in $S_0$. From (\ref{dd0}) we see that for any edge $e=\langle p,q \rangle\in \ell_{uv}$, $p$ and $q$ can be joined by a path $\tau_{pq}$ in $G_0$ whose length does not exceed $D_0$. By replacing each edge $\langle p,q\rangle$ by the path $\tau_{pq}$ in $\ell_{uv}$, we obtain a path $\ell_{pq}^0$ in $G_0$ joining $u$ and $v$ and visiting no vertices in $F_1\cup F_2$ by the definition of $S_0$ in (\ref{s0}). This is equivalent to the condition that there exists a path in $G_0$ joining $1_{\Gamma}$ and $u^{-1}v$, which visits no vertices in $u^{-1}[F_1\cup F_2]$. But this is a contradiction to Lemma \ref{css}. 
 
For $u$ and $v$ satisfying the condition of the theorem, assume the lengths of the normal forms of $u$ and $v$ are strictly greater than the maximal lengths of normal forms of elements in $S_0$. Assume the normal form of $u$ (resp.\ $v$) has length $n_1$ (resp.\ $n_2$). Let $\epsilon_{n_1}(u)$ (resp.\ $\epsilon_{n_2}(v)$) be the exponent of the $n_1$th (resp.\ $n_2$th) $t$ in the normal form of $u$ (resp.\ $v$). Then $\{u t^{\epsilon_{n_1}(u)k}\}_{k=1}^{\infty}$ and $\{v t^{\epsilon_{n_2}(v)k}\}_{k=1}^{\infty}$ are in two distinct infinite components of $G\setminus S_0$.
\end{proof}

Let $S_1$ be a finite, connected set of vertices in $G$ containing $S_0$.  By Lemmas \ref{l41} and \ref{l48}, the fact that $G\setminus S_0$ has at least two distinct infinite components implies that $G\setminus S$ has at least two distinct infinite components. 

\subsubsection{$F_1=F_2=H$}

In this section, we consider the case when the group $\Gamma$ is an HNN extension, as in Definition \ref{dfn}, with $F_1=F_2=H$.

\begin{lemma}Assume that the group $\Gamma$ is a finitely generated HNN extension, as in Definition \ref{dfn}, such that $F_1=F_2=H$. Then any locally finite Cayley graph $G$ of $\Gamma$ has two ends.
\end{lemma}
\begin{proof}It is a well-known fact that the number of ends of locally finite Cayley graphs of a finitely generated group do not depend on the choices of finite generating set. Therefore it suffices to prove the theorem for a specific choice of generating set. Let $G$ be the Cayley graph with respect to the generating set $H\cup\{t,t^{-1}\}$. It is straightforward to check that $G$ has two ends.
\end{proof}

\subsubsection{$|F_1|=|F_2|<|H|$}\label{sl}

Now we prove Theorem \ref{mg} when the group $\Gamma$ is a finitely generated HNN extension, as in Definition \ref{dfn}, such that $F_1$ is a proper subset of $H$. Since $F_1$ is finite and $F_1$ and $F_2$ are isomorphic groups, this implies that $|F_1|=|F_2|<|H|$, and therefore $F_2$ is also a proper subset of $H$.

\begin{lemma}When the group $\Gamma$ is a finitely generated HNN extension, as in Definition \ref{dfn}, such that $F_1$ is a proper subset of $H$, then there exists $a\in(0,1]$, such that
\begin{eqnarray*}
\limsup_{n\rightarrow\infty}|\{\pi_n:\|\pi_n\|\geq an\}|^{\frac{1}{n}}=\mu,
\end{eqnarray*}
where $\pi_n$ is an $n$-step SAW starting from $1_{\Gamma}$, and $\mu$ is the connective constant.
\end{lemma}
\begin{proof}It suffices to check Assumption \ref{ap31}, then the lemma follows from Theorem \ref{m31} A.

Let $S_0$ be defined as in (\ref{s0}). Let $S$ be a finite set of vertices including $S_0$ such that $S$ is connected.

Let $v\in S$ be incident to $w\in G\setminus S$. Let $A_w$ be the component of $G\setminus S$ including $w$. Let $k$ be an integer which is strictly greater than the maximal length of normal forms of elements in $S$. The we define the graph automorphism $\phi(S,A_{w})$ in Assumption \ref{ap31} according to the normal form of $w$ as follows.
\begin{Alist}
\item If $w$ has a normal form with $g_0\in H\setminus [F_2\cup F_1]$, let $\phi(S,A_{w})=t^{-2k}$;
\item If $w$ has a normal form with $g_0\in F_1$, $\epsilon_1=1$, let $\phi(S,A_w)= t^{-2k}r_1$, where $r_1\in H\setminus F_1$;
\item If $w$ has a normal form with $g_0\in F_1\setminus F_2$, $\epsilon_1=-1$, and $A_{w}$ contains no vertices satisfying Case A.; let $\phi(S,A_w)=t^{2k}$;
\item If $w$ has a normal form with $g_0\in F_2$, $\epsilon_1=-1$, let $\phi(S,A_w)=r_2t^{-2k}$, where $r_2\in H\setminus F_2$;
\item If $w$ has a normal form with $g_0\in F_2\setminus F_1$, $\epsilon_1=1$, and $A_{w}$ contains no vertices satisfying Case A. or Case C.; let $\phi(S,A_w)=t^{-2k}$.
\end{Alist}

\begin{lemma}\label{ll1}In Case A., $\phi(S,A_w)A_w$ is in a component of $G\setminus S$ different from $A_w$.
\end{lemma}
\begin{proof}Assume Case A. occurs.
 Let $u$ be an arbitrary vertex in $A_w$, then one of  the following 3 cases must occur
\begin{romlist}
\item $u$ has a normal form with $g_0\in H\setminus [F_2\cup F_1]$
\item $u$ has a normal form with $g_0\in F_1\setminus F_2$, $\epsilon_1=-1$;
\item $u$ has a normal form with $g_0\in F_2\setminus F_1$, $\epsilon_1=1$.
\end{romlist}
To see why that is true, let us look at Cases (a)-(d) of Lemma \ref{css}. First of all, note that the vertex $w\in A_w$ satisfies Case (i), therefore $A_w$ does contain vertices in Case (i). 

Now we consider vertices in $A_w$ whose normal form has $g_0\in [F_1\cup F_2]$. The following cases might occur
\begin{itemize}
\item If $w$ has a normal form with $\epsilon_1=-1$, then by Cases (a) and (d) of Lemma \ref{css} and Lemma \ref{l48}, for any vertex $u$ in $A_w$ whose normal form has $g_0\in [F_1\cup F_2]$, either $u$ has a normal form with $\epsilon_1=1$ and $g_0\in F_2\setminus F_1$, or $u$ has a normal form with $\epsilon_1=-1$ and $g_0\in F_1\setminus F_2$, because otherwise $u$ and $w$ are in two distinct components of $G\setminus S$; or one of them is in $S$.
 \item If $w$ has a normal form with $\epsilon_1=1$, then by Cases (b) and (c) of Lemma \ref{css} and Lemma \ref{l48}, for any vertex $u$ in $A_w$ whose normal form has $g_0\in [F_1\cup F_2]$, either $u$ has a normal form with $\epsilon_1=1$ and $g_0\in F_2\setminus F_1$, or $u$ has a normal form with $\epsilon_1=-1$ and $g_0\in F_1\setminus F_2$, because otherwise $u$ and $w$ are in two distinct components of $G\setminus S$; or one of them is in $S$.
 \end{itemize}

For Case (i), $\phi(S,A_{w})u$ and $u$ are in distinct components of $G\setminus S$ by  (b)(d) of Lemma \ref{css}. For all the Cases (i) (ii) and (iii),
if $\phi(S,A_w)u\in S$, then $u\in t^{2k} S$, then $u$ has normal form with $g_0\in F_1$, $\epsilon_1=1$, but this is possible in none of Cases (i), (ii) and (iii). Therefore $\phi(S,A_w)A_w$ is in a component of $G\setminus S$ different from $A_w$.
\end{proof}

\begin{lemma}\label{ll2}In Case B., $\phi(S,A_w)A_w$ is in a component of $G\setminus S$ different from $A_w$.
\end{lemma}

\begin{proof}
For Case B. $\phi(S,A_{w})w$ and $w$ are in distinct components of $G\setminus S$ by (b) of Lemma \ref{css}. Note that any vertex in $A_w$ has a normal form with $g_0\in F_1$, $\epsilon_1=1$ by (a) (c) of Lemma \ref{css}. Therefore $\phi(S,A_w)A_w$ is in a component of $G\setminus S$ different from $A_w$.
\end{proof}

\begin{lemma}\label{ll3}In Case C., $\phi(S,A_w)A_w$ is in a component of $G\setminus S$ different from $A_w$.
\end{lemma}

\begin{proof}
For Case C. $\phi(S,A_{w})w$ and $w$ are in distinct components of $G\setminus S$ by (a) of Lemma \ref{css}. Let $u$ be an arbitrary vertex in $A_w$, then by Cases (a) (d) of Lemma \ref{css} $u$ must have a normal form satisfying one of the following two conditions
\begin{enumerate}
\item
 $g_0\in H\setminus F_1$, $\epsilon_1=1$; or  
\item $g_0\in H\setminus F_2$, $\epsilon_1=-1$.
 \end{enumerate}
 If $\phi(S,A_{w})u\in S$, then $u\in t^{-2k}S$.
 But this is not possible since in this case the normal form of $u$ satisfies $\epsilon_1=-1$ and $g_0\in F_2$. Therefore $\phi(S,A_w)A_w$ is in a component of $G\setminus S$ different from $A_w$.
\end{proof}

\begin{lemma}\label{ll4}In Case D., $\phi(S,A_w)A_w$ is in a component of $G\setminus S$ different from $A_w$.
\end{lemma}

\begin{proof}
For Case D. $\phi(S,A_{w})w$ and $w$ are in distinct components of $G\setminus S$ by (d) of Lemma \ref{css}. Note that any vertex in $A_w$ has a normal form with $g_0\in F_2$, $\epsilon_1=-1$ by Cases (b) (d) of Lemma \ref{css}. Then if $\phi(S,A_w)u\in S$, $u\in t^{2k} r_2^{-1}S$ has a normal form with $\epsilon_1=1$, which is impossible.
 Therefore $\phi(S,A_w)A_w$ is in a component of $G\setminus S$ different from $A_w$.
\end{proof}

\begin{lemma}\label{ll5}In Case E., $\phi(S,A_w)A_w$ is in a component of $G\setminus S$ different from $A_w$.
\end{lemma}

\begin{proof}
For Case E. $\phi(S,A_{w})w$ and $w$ are in distinct components of $G\setminus S$ by (b) of Lemma \ref{css}. By (b) (c) of Lemma \ref{css},  any vertex $u$ in $A_w$ has a normal form satisfying one of the following conditions
\begin{enumerate}
\item $g_0\in H\setminus F_1$, $\epsilon_1=1$; or
\item $g_0\in H\setminus F_2$, $\epsilon_1=-1$.
 \end{enumerate}
If $\phi(S,A_w)u\in S$, then $u\in t^{2k}S$. Then $u$ has a normal form with $\epsilon_1=1$ and $g_0=1\in F_1$, but this is a contradiction to Cases (1) and (2). Therefore $\phi(S,A_w)A_w$ is in a component of $G\setminus S$ different from $A_w$.
\end{proof}

Let $B_w$ be the component of $G\setminus S$ containing $\phi(S,A_w)A_w$. By the construction above, we have
\begin{eqnarray*}
\phi(S,A_w)S\cap S=\emptyset,
\end{eqnarray*}
since for each element in $\phi(S,A_w)S$, its normal form has a length strictly greater than the length of the normal form of any element in $S$.
Therefore, $\phi(S,A_w)A_w\subset B_w$ implies
\begin{eqnarray}
\phi(S,A_w) S\subset B_{w}.\label{bw}
\end{eqnarray}
\begin{lemma}\label{l424}
In Cases A.-E. that $A_w$ and $\phi(S,A_w)A_w$ are in two distinct components of $G\setminus[\phi(S,A_w)S]$. 
\end{lemma}

\begin{proof}It suffices to show that in Cases A.-E., $A_w$ and $\phi(S,A_w)^{-1}A_w$ are in two distinct components of $G\setminus S$. Then the lemma follows from similar arguments as in the proofs of Lemmas \ref{ll1}-\ref{ll5}.
\end{proof}
For each $v\in \partial_{A_w}S$, we can construct a path $\ell_v$ joining $v$ and $\phi(S,A_w)v$ as follows
\begin{numlist}
\item use a path $\ell_1$ in $S$ to join $v$ and a vertex $u$ in $\partial_{B_w}S$ - this is possible by the connectivity of $S$;
\item use a shortest path $\ell_2$ in $B_w$ to join $u$ and $\phi(S,A_w)S$; let $x$ be the endpoint of $\ell_2$;
\item use a path $\ell_3$ in $\phi(S,A_w)S$ to join $x$ and $\phi(S,A_w)v$.
\end{numlist}
Let $\ell_v$ be the concatenation of $\ell_1$, $\ell_2$, and $\ell_3$. 

\begin{lemma}
$\ell_v\in G\setminus [A_w\cup \phi(S,A_w)A_w]$.
\end{lemma}
\begin{proof}Since $\ell_v=\ell_1\cup\ell_2\cup\ell_3$, it suffices to show that for $1\leq i\leq 3$, $\ell_i\in G\setminus [A_w\cup \phi(S,A_w)A_w]$.

 The path $\ell_1\subset S$, $S\cap A_w=\emptyset$ and $S\cap \phi(S,A_w)A_w=\emptyset$. Therefore $\ell_1\in G\setminus [A_w\cup \phi(S,A_w)A_w]$. 
 
 The path $\ell_3\in \phi(S,A_w)S$; $\phi(S,A_w)S\cap \phi(S,A_w)A_w=\emptyset$ since $S\cap A_w=\emptyset$; $\phi(S,A_w)S\cap A_w=\emptyset$ by (\ref{bw}) and the fact that $B_w\cap A_w=\emptyset$. Therefore $\ell_3\in G\setminus [A_w\cup \phi(S,A_w)A_w]$. 
 
 The path $\ell_2\in B_w$, hence $\ell_2\cap A_w=\emptyset$. Since $[\ell_1\cup\ell_2\setminus\{x\}]\cap \phi(S,A_w)S=\emptyset$; and $v$ is in the same component of $G\setminus[\phi(S,A_w)S]$ as $A_w$, we obtain that $\ell_1\cup\ell_2\setminus\{x\}$ are in in the same component of $G\setminus[\phi(S,A_w)S]$ as $A_w$. By Lemma \ref{l424}, we have  $[\ell_1\cup\ell_2\setminus\{x\}]\cap \phi(S,A_w)A_w=\emptyset$. Therefore $\ell_2\in G\setminus [A_w\cup \phi(S,A_w)A_w]$. 
\end{proof}

The lengths of $\ell_1$ and $\ell_3$ are bounded above by $|S|$. We can make the length of $\ell_2$ to be bounded above by the distance of $\partial_{B_w}S$ and $\phi(S,A_w)S$, which is bounded above by the graph distance in $G$ if $1_{\Gamma}$ and $\phi(S,A_w)1_{\Gamma}$. The latter is bounded by $2k+\max\{\mathrm{dist}_{G}(1_{\Gamma},r_1),\mathrm{dist}_{G}(1_{\Gamma},r_2)\}$. Hence if we choose
\begin{eqnarray*}
N=2|S|+2k+\max\{\mathrm{dist}_{G}(1_{\Gamma},r_1),\mathrm{dist}_{G}(1_{\Gamma},r_2)\},
\end{eqnarray*}
then Assumption \ref{ap31}(3) is satisfied.

Therefore Theorem \ref{mg} when $\Gamma$ is an HNN extension with $|F_1|=|F_2|<H$ follows from Theorem \ref{m31} A.
\end{proof}

\subsubsection{Proof of Theorem \ref{t15} when $\Gamma$ is an HNN extension}

Let $\Gamma$ be an infinite, finitely generated graph which is an HNN extension as described by (2). It suffices to construct a locally finite Cayley graph $G$ of $\Gamma$ on which SAWs have positive speed.

First we consider the case when $F_1=F_2=H$. Let $G$ be the Cayley graph of of $\Gamma$ with respect to  generator set $H\cup\{t,t^{-1}\}$; i.e. any elements in $H$ corresponds to an edge in $G$.
Let $S'=S=H$. Note that $\phi\setminus S$ has two distinct infinite components. For any component $A$ of $G\setminus S$, let $\phi(S,A)$ be the mapping from $\Gamma$ to $\Gamma$ changing each $t$ in the normal form to $t^{-1}$ and each $t^{-1}$ in the normal form to $t$.
 Then Theorem \ref{t15} in this case follows from Theorem \ref{m31} B.

Now we consider the case when $F_1$ is a proper subset of $H$. Let $T_H$ be a finite generator set of $H$ satisfying $T_H=T_H^{-1}$, $1\notin T_H$, $|T_H|<\infty$, and $[F_1\cup F_2]\setminus \{1_{\Gamma}\}\subset T_H$. Let $G$ be a Cayley graph of $\Gamma$ with respect to the set of generators $T_H\cup\{t,t^{-1}\}$. Let $S$ be defined as in Section \ref{sl}, and let $S'=S$.
Then Theorem \ref{t15} in this case follows from Theorem \ref{m31} B.

\section{Free product graph of two quasi-transitive graphs}\label{fp}

In this section, we prove Theorem \ref{tm51}.

\begin{proof}Obviously $G$ is an infinite, connected, quasi-transitive graph. Let $S=\{o\}\in V$. Then $G\setminus S$ has at least two infinite components. Indeed, Let $x,y\in V$ satisfy
\begin{eqnarray}
x&=&x_1\ldots x_n;\label{x}\\
y&=&y_1\ldots y_m; \label{y}
\end{eqnarray}
where $m,n\geq 1$, $x_i,y_j\in V_1^{\times} \cup V_2^{\times}$, $I(x_i)\neq I(x_{i+1})$, $I(y_j)\neq I(y_{j+1})$ (see Definition \ref{df41} for notations). If $x_1\in V_1^{\times}$ and $y_1\in V_2^{\times}$, then $x$ and $y$ are in two distinct components of $G\setminus S$. 

Let $A$ (resp.\ $B$) be a component of $G\setminus S$, such that for any $x\in A$ (resp.\ $y\in B$), $x$ (resp.\ $y$) has the form (\ref{x}) (resp. (\ref{y})) with $x_1\in V_1^{\times}$ (resp. $y_1\in V_2^{\times}$). Let $u\in V_2^{\times}$ (resp.\ $w\in V_2^{\times}$), and define $\phi(S,A)x=ux$ (resp. $\phi(S,B)=vy$). Then it is straightforward to verify Assumption \ref{ap32} with $S$ chosen as above. Therefore the theorem follows from Theorem \ref{t15}.
\end{proof}

\bigskip
\noindent{\textbf{Acknowledgements.}} The author thanks Yuval Peres, Geoffrey Grimmett for helpful discussions. The author's research is partially supported by National Science Foundation grant $\#$1608896.

\bibliography{new5}
\bibliographystyle{amsplain}

\end{document}